\documentclass{amsart}

\usepackage[utf8]{inputenc}
\usepackage[T1]{fontenc}
\usepackage{amsmath,amssymb,amsthm}
\usepackage[svgnames]{xcolor}
\usepackage{graphicx}
\usepackage{doi}
\usepackage{hyperref}
\usepackage{csquotes}
\usepackage[inline]{enumitem}
\usepackage{booktabs}
\usepackage{microtype}
\usepackage{tikz}
\usepackage{tkz-graph}
\usepackage[square,numbers]{natbib} 
\usepackage{todonotes}
\usepackage[capitalize]{cleveref}
\usepackage[notcite,notref,color,final]{showkeys}

\setlist[enumerate]{label=\textup{(\roman{*})}}
\usetikzlibrary{arrows,positioning,cd} 

\newcommand*{\N}{\mathbb{N}} 
\newcommand*{\nn}{\mathbb{N}} 
\newcommand*{\zz}{\mathbb{Z}} 
\newcommand*{\Z}{\mathbb{Z}}

\newcommand*{\ret}{\mathcal{R}}
\newcommand*{\lang}{\mathcal{L}}
\newcommand*{\cA}{\mathcal{A}}
\newcommand*{\cB}{\mathcal{B}}
\newcommand*{\cC}{\mathcal{C}}
 
\newcommand*{\from}{\colon} 
\newcommand*{\id}{\mathrm{id}}  
\newcommand*{\emptyw}{\varepsilon}
 
\newcommand*{\lcm}{\operatorname{lcm}}
\newcommand*{\Gr}{\operatorname{Gr}} 

\newcommand*{\leftExt}[1]{\operatorname{E}^-_{#1}}
\newcommand*{\rightExt}[1]{\operatorname{E}^+_{#1}} 
\newcommand*{\ExtGraph}[1]{\mathcal{E}_{#1}}
\newcommand*{\RauzyGraph}{\Gamma}
\newcommand*{\ab}{\operatorname{ab}}

\newcommand{\displaysub}[1]{\begin{cases}#1\end{cases}}

\newcommand*{\card}{\#}

\newcommand*{\gen}[1]{\left\langle #1\right\rangle}

\theoremstyle{plain}
\newtheorem{theorem}{Theorem}
\newtheorem{proposition}[theorem]{Proposition}
\newtheorem{lemma}[theorem]{Lemma}
\newtheorem{corollary}[theorem]{Corollary}

\theoremstyle{definition}
\newtheorem{definition}[theorem]{Definition}

\theoremstyle{remark} 
\newtheorem{remark}[theorem]{Remark}
\newtheorem{example}[theorem]{Example}

\begin{document}

\author[F. Gheeraert]{France Gheeraert}
\author[H. Goulet-Ouellet]{Herman Goulet-Ouellet}
\thanks{The second author was supported by the CTU Global Postdoc Fellowship program.}
\author[J. Leroy]{Julien Leroy}
\author[P. Stas]{Pierre Stas}

\title{Stability properties for subgroups generated by return words}

\begin{abstract}
Return words are a classical tool for studying shift spaces with low factor complexity. In recent years, their projection inside groups have attracted some attention, for instance in the context of dendric shift spaces, of generation of pseudorandom numbers (through the welldoc property), and of profinite invariants of shift spaces. Aiming at unifying disparate works, we introduce a notion of stability for subgroups generated by return words. Within this framework, we revisit several existing results and generalize some of them. We also study general aspects of stability, such as decidability or closure under certain operations. 
\end{abstract}

\maketitle

\section{Introduction}

Introduced in the 90's by Durand~\cite{Durand1998}, return words find their origin in the induction operation on shift spaces. They are roughly defined as words separating two consecutive occurrences of a given word. This notion is used to construct $S$-adic representations~\cite{Durand1998}, to characterize families of shift spaces~\cite{Vuillon2001,Balkova2008,pre/Gheeraert2024}, to obtain decision procedure for substitutive shift spaces~\cite{Durand_minimality,Durand_Leroy_2022}, and to study bifix codes~\cite{maximal_bifix_decoding}. Return words also have other applications in symbolic dynamics, combinatorics on words, and algebra: they play a key role in the study of Schützenberger groups, which are profinite groups naturally associated with minimal shift spaces (see \cite{book/Almeida2020}), and they appear when studying critical exponents~\cite{Dolce2023,Dvorakova2024}, and palindromic richness~\cite{Glen2009b}. 

It is natural to consider words as elements of the free group, and more generally to consider their projection into some group, such as the free Abelian group, or groups of the form $(\zz/m\zz)^d$, $m\geq 2$. Several papers studying the behavior of return groups, i.e., subgroups generated by return words in a given group, can be found in the literature; we briefly survey some of them.

Return groups were first described in the study of Sturmian and strict episturmian shift spaces, where return words to any $w$ form a basis of the free group over the alphabet (\cite{Justin2000,Berstel2012}). This is partially generalized in~\cite{Berthe2015} by Berthé et al. to a much wider class called dendric shift spaces, which includes both strict episturmian words and codings of regular interval exchanges. Indeed, in a dendric shift space, return sets form bases of the ambient free group, a result called the \emph{Return Theorem}. This condition is in fact a characterization, as shown in~\cite{pre/Gheeraert2024}. Berthé et al. also observed that the return words still form generating sets (but need not form bases) under the weaker assumption of \emph{connectedness}; and in fact the even weaker assumption of suffix-connectedness suffices~\cite{Goulet-Ouellet2022}. 
An analogous property in finite Abelian groups appears in the work of Balková et al. on pseudorandom number generators~\cite{Balkova2016}, in the form of the \emph{welldoc property}. The general behavior of return groups in finite groups is also linked with ergodic properties of skew products based on minimal shift spaces~\cite{pre/Berthe2024}. Another work reveals that return groups in the Thue--Morse shift space have a peculiar behavior: they form infinite decreasing chains~\cite{Berthe2023}.

Our aim is to initiate a systematic study of return groups in order to provide a coherent explanation for these disparate results. Indirectly, all of them study whether or not all return groups are equal, at least for sufficiently large words. This can be formalized using \emph{eventual stability} defined as follows: we ask that the return group to a sufficiently long word $u$ is equal to the return group to its prefix of some length bounded independently of $u$. The smallest such bound is called the \emph{stability threshold}, and we also study \emph{stability} defined as eventual stability of threshold 0. To emphasize the fact that (eventual) stability depends on the ambient group, we talk about (eventual) $\varphi$-stability, where $\varphi$ is a morphism used to map the return words to the ambient group.

In this paper, we revisit the results mentioned above through the lens of (eventual) stability, we generalize some of them, and we study several new aspects of stability such as decidability or operations on shift spaces. More precisely, the paper is organized as follows.

The main notions from symbolic dynamics are recalled in~\cref{s:pre} while eventual stability is introduced in~\cref{S:stability and fullness}. We then study stability for the above-mentioned families of shift spaces. In particular, using a key theorem by Krawczyk and Müllner~\cite{these_Krawczyk}, we show that no aperiodic automatic shift space is eventually stable (\cref{c:automatic}). The end of~\cref{S:stability and fullness} provides several equivalent definitions of eventual stability: using Rauzy graphs (\cref{P:equivalence with Rauzy groups}) or in subgroup separable groups (\cref{thm:main}).

In \cref{S:decidability}, we establish some decidability results in the context of shift spaces generated by a primitive substitution. First, the \emph{limit} of the return groups read in a finite group can be effectively computed (\cref{P:H computable}). Moreover, at the cost of imposing additional conditions on the substitution, we obtain analogous results when return groups are read in the free group over the alphabet (\cref{P:decidable for derivating and bifix}).

\cref{S:closure properties} is dedicated to the study of the behavior of (eventual) stability under two natural operations on shift spaces: derivation and application of a substitution. For both, (eventual) stability is preserved whenever the group morphisms involved satisfy some mild conditions (\cref{P:stable preserved by derivation}, \cref{P:eventual stability preserved by substitution}, and \cref{P:stability preserved by substitution}). A consequence of these results is that in the free group, eventual stability is preserved by application of substitutions (\cref{c:closed under substitution}). This cannot be said about derivation however; we provide a counterexample (\cref{p:contre-exemple}). 

Afterwards, we show that stability in subgroup separable groups respects a local-global principle, in the sense that it is sufficient to consider stability in the finite quotients (\cref{p:local-global}). This is the object of \cref{S:local-global}. Using the local-global principle, we study the welldoc property in \cref{s:welldoc}. We restate it in terms of Abelian stability (\cref{t:abelian-welldoc}) and, as a consequence of this work, show that it is satisfied by a large family of shift spaces (\cref{c:welldoc}).

\section{Preliminaries}
\label{s:pre}

Let $\cA$ be a finite alphabet. The symbol $\varepsilon$ denotes the \emph{empty word}, and $\cA^*$ is the \emph{free monoid} over $\cA$. We also let $\cA^+ = \cA^*\setminus\{\varepsilon\}$ be the set of all non-empty words over $\cA$. For $u \in \cA^*$, $|u|$ denotes the length of $u$.

Let $\cA^\zz$ and $\cA^\nn$ be respectively the sets of two-sided and right-sided infinite words over $\cA$. Given $x\in\cA^\zz\cup\cA^\nn\cup\cA^*$, we let $\lang(x)\subseteq\cA^*$ denote the \emph{language} of $x$, i.e., the set of all finite words appearing in $x$. We extend this notation to subsets $X\subseteq \cA^\zz$ by $\lang(X) = \bigcup_{x\in X}\lang(x)$. We further write
\begin{gather*}
    \lang_{\le n}(X) = \{ u \in \lang(X) \mid |u| \le n \},\quad \lang_{\ge n}(X) = \{ u \in \lang(X) \mid |u| \ge n \},\\ \lang_n(X) = \{ u \in \lang(X) \mid |u| = n \}.
\end{gather*}

\subsection{Shift spaces}

Let us equip the sets $\cA^\nn$ and $\cA^\zz$ with the product topology, where $\cA$ is equipped with the discrete topology. The \emph{shift} is the map $S \from \cA^\zz \to \cA^\zz$ defined by $(x_n)_n \mapsto (x_{n+1})_n$. 

A non-empty subset $X$ of $\cA^\zz$ is a \emph{shift space} if it is closed and invariant under the shift map. The shift space $X$ is \emph{minimal} if it does not strictly contain any other shift space. It should be noted that this condition is equivalent to $\lang(X)$ being \emph{uniformly recurrent}, which is to say that, for any $w \in \lang(X)$, there is a $k$ such that any $u \in \lang_{\ge k}(X)$ contains $w$ as a factor. 

A shift space is called \emph{periodic} if it contains only purely periodic infinite words, and \emph{aperiodic} if it contains no such words. Note that a minimal shift space is either periodic or aperiodic, and that the latter is equivalent to the shift space being infinite.

\subsection{Substitutions}

Since we also work with group morphisms, to avoid any confusion, we call \emph{substitution} a morphism between finitely generated free monoids, i.e., a monoid morphism of the form $\sigma\from \cA^* \to \cB^*$. Such a morphism $\sigma$ may be extended, whenever convenient, to maps $\cA^\zz\to \cB^\zz$ and $\cA^\nn\to\cB^\nn$ which are again denoted $\sigma$. We say that $\sigma$ is \emph{letter-to-letter} whenever $\sigma(a)$ is a letter for all $a \in \cA$.

A sequence $x \in \cA^\nn$ is said to be {\em substitutive} if there exist substitutions $\sigma\from\cB^* \to \cB^*$ and $\theta\from \cB^* \to \cA^*$ and a letter $b \in \cB$ such that $\sigma^n(b)$ converges in $\cB^\nn$ to $y$ (which is a fixed point of $\sigma$) and $x = \theta(y)$.
In turn, a shift space $X$ over $\cA$ is called \emph{substitutive} if it is generated by a substitutive sequence $x$, i.e.,
\[
    X = \left\{z \in \cA^\zz \mid \lang(z) \subseteq \lang(x) \right\}. 
\]
We also say that $X$ is \emph{generated by} $(\sigma,\theta)$.
Whenever $x$ is substitutive and $\theta$ is the identity, we say that $x$ (as well as the shift space it generates) is {\em purely substitutive}.

In what follows, we mostly deal with substitutive shifts spaces where the substitution $\theta$ is letter-to-letter and the substitution $\sigma$ is {\em primitive}, that is, there exists a positive integer $k$ such that for all $a,b \in \cA$, the letter $a$ occurs in $\sigma^k(b)$.
In that case, the generated shift space $X$ is minimal and independent of the choice of the fixed point $y$ of $\sigma$. More precisely, we have
\[
    X = \left\{z \in \cA^\zz \mid \lang(z) \subseteq \bigcup_{n \in \nn, b \in \cB}\lang(\theta\sigma^n(b))\right\}.
\]
We say that such shift spaces are {\em primitive substitutive}, and every minimal substitutive shift space is primitive substitutive~\cite{Durand1998}.

\subsection{Return words and return groups}

For a shift space $X\subseteq\cA^\Z$ and a word $u \in \lang(X)$, a \emph{return word} to $u$ is an element of the set
\begin{equation*}
    \ret_X(u) = \left\{ r\in \cA^+ \mid ru\in \lang(X)\cap (u\cA^*\setminus \cA^+u\cA^+)\right\}.
\end{equation*}
Note that $\ret_X(\emptyw) = \cA$ and if $u\neq\emptyw$, the return words to $u$ are the words $r$ such that $ru \in \lang(X)$ and $ru$ contains exactly two occurrences of $u$, one as a prefix and one as a suffix.

\begin{example}\label{eg:Tribo}
    Take the so-called \emph{Tribonacci substitution} $\sigma\from \{a,b,c\}^*\to \{a,b,c\}^*$ defined by 
    \begin{equation*}
        \sigma\from a\mapsto ab, b\mapsto ac, c\mapsto a.
    \end{equation*}
    Let $X$ be the shift space generated by $\sigma$. The set of return words to $aba$ in $X$ is given by $\ret_X(aba) = \{ab, aba, abac\}$.
    All three of these return words can be seen in the factor 
    \[
    \sigma^5(a) = \underline{abac}abaabac\underline{ab}abac\underline{aba}abac.
    \]
\end{example}

For an alphabet $\cA$, let $F_\cA$ denote the \emph{free group} over $\cA$. Recall that it consists of all reduced words over the alphabet $\cA\cup \cA^{-1}$, where $\cA^{-1} = \{a^{-1} \mid a\in \cA\}$ is a disjoint copy of $\cA$, and a word is \emph{reduced} when it contains no factors of the form $aa^{-1}$ or $a^{-1}a$ with $a\in \cA$. The cardinality of $\cA$ is called the \emph{rank} of $F_\cA$. For a general reference on free groups and their basic properties, we refer to the monograph~\cite{Lyndon2001}. We also recall, for context, the foundational Nielsen--Schreier theorem, which states that any subgroup of a free group is isomorphic to a free group (of a possibly different rank). 

Given a group $G$ and a subset $T \subseteq G$, we let $\gen{T}$ denote the subgroup of $G$ generated by $T$, i.e., the smallest subgroup of $G$ containing $T$ (where the dependency on $G$ is implicit). We view $\cA^*$ as a subset of $F_\cA$ in the usual way, so we naturally extend any substitution $\sigma\from \cA^* \to \cB^*$ into a group morphism from $F_\cA$ to $F_\cB$. We also consider subgroups of $F_\cA$ generated by words of $\cA^*$, as in the following central notion.
\begin{definition}
    Let $X$ be a shift space over $\cA$ and let $u \in \lang(X)$. The \emph{return group to $u$} in $X$ is the subgroup $\gen{\ret_X(u)}$ of $F_\cA$.
\end{definition}

More generally, we are interested in the study of images of return groups under group morphisms $\varphi\from F_\cA\to G$ where the image $\varphi(F_\cA)$ often has additional properties (such as being finite). In this case, it will be more convenient to assume that $\varphi$ is onto and that $G$ itself satisfies these properties. This is not a restriction since one can simply replace $G$ by $\varphi(F_\cA)$ if needed. The following simple lemma is used repeatedly.

\begin{lemma}\label{lem:return groups of factors}
    Let $X$ be a minimal shift space over $\cA$ and let $\varphi\from F_\cA\to G$ be any group morphism. If $w$ is a factor of $u \in \lang(X)$, then up to conjugacy, $\varphi\gen{\ret_X(u)}$ is a subgroup of $\varphi\gen{\ret_X(w)}$. More precisely, if $pw$ is a prefix of $u$, then
    \begin{equation*}
        \varphi\gen{\ret_X(w)} \geq \varphi(p)^{-1}\varphi\gen{\ret_X(u)}\varphi(p).
    \end{equation*}
    In particular, if $w$ is a prefix of $u$, then $\varphi\gen{\ret_X(w)} \geq \varphi\gen{\ret_X(u)}$.
\end{lemma}

\begin{proof}
    Assume that $u = pws$. First observe that any return word to $u = pws$ is a concatenation of return words to $pw$ (\cite[Proposition~2.6, part~4]{Durand1998}). Likewise, any return word to $pw$ is, up to conjugacy by $p$, a concatenation of return words to $w$. Therefore
    \begin{equation*}
        \gen{\ret_X(w)} \geq p^{-1}\gen{\ret_X(pw)} p \geq p^{-1}\gen{\ret_X(u)} p.
    \end{equation*}
    As $\varphi$ is a morphism, we further deduce that $\varphi\gen{\ret_X(w)} \geq \varphi(p)^{-1}\varphi\gen{\ret_X(u)}\varphi(p)$.
\end{proof}

\section{Eventual stability}\label{S:stability and fullness}

\cref{lem:return groups of factors} shows that return groups are subgroups of one another when extending words on the right. We therefore consider the following version of stability.

\begin{definition}\label{D:stable}
    Let $\varphi\from X\to G$ be any group morphism. A shift space $X$ over $\cA$ has \emph{eventually $\varphi$-stable return groups} if there exists $M \in\nn$ such that, for all $w \in \lang_M(X)$ and all $u \in w\cA^* \cap \lang(X)$, we have $\varphi\gen{\ret_X(u)} = \varphi\gen{\ret_X(w)}$. If $M = 0$, we moreover say that $X$ has \emph{$\varphi$-stable return groups}. 
\end{definition}

We make two comments on this definition. First, the set of integers $M$ satisfying the above property is either empty or an interval of the form $[K,\infty)$. In the latter case, the integer $K$ is called the \emph{threshold}. Second, since $\ret_X(\emptyw) = \cA$, one could equivalently say that $X$ has $\varphi$-stable return groups whenever $\varphi\gen{\ret_X(u)} = \varphi(F_\cA)$ for all $u \in \lang(X)$.

Eventual stability is of particular interest when $\varphi$ is the identity on $F_\cA$, in which case we omit the mention of $\varphi$ and say that $X$ has \emph{(eventually) stable return groups}. This omission is justified by the following simple result.

\begin{proposition}\label{P:free group universal}
    Let $X$ be a minimal shift space over $\cA$ with eventually stable return groups of threshold $M$. Then $X$ has eventually $\varphi$-stable return groups of threshold at most $M$ for every group morphism $\varphi$ defined on $F_\cA$.
\end{proposition}

Eventual stability can also be stated in terms of groups generated by Rauzy graphs. Let $X$ be a shift space. 
The \emph{order-$n$ Rauzy graph} of $X$, denoted $\RauzyGraph_n(X)$, is the directed graph whose vertices are the elements of $\lang_n(X)$, and there is an edge from $u$ to $v$ if there are letters $a$ and $b$ such that $ub = av \in \lang(X)$ ; this edge is labeled by $a$. The label of a path is the concatenation of the labels of the edges.

The \emph{Rauzy group} of $u \in \lang(X)$, denoted $\Gr(u)$, is the subgroup of $F_\cA$ generated by the labels of paths based on $u$ in $\RauzyGraph_{|u|}(X)$, where by {\em path based on $u$} we mean a path from $u$ to $u$.

\begin{example}
    Let us give an explicit example of a Rauzy graph. We once again use the shift space generated by the Tribonacci substitution introduced in \cref{eg:Tribo}. Its order-2 Rauzy graph is computed as follows: one can check that $\lang_3(X) = \{aab,aba,aca,baa,bab,bac,cab\}$, and therefore that $\lang_2(X) = \{aa,ab,ac,ba,ca\}$. As such, we obtain the order-2 Rauzy graph represented below. 
    
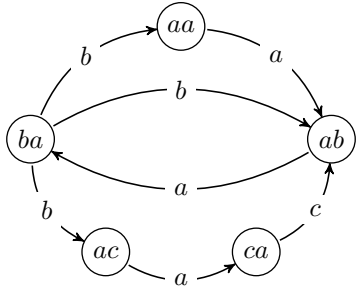
\begin{center}
    \begin{tikzpicture}
            \Vertex[style={shape=circle},L=$aa$,x=0cm,y=1.5cm]{v00}
            \Vertex[style={shape=circle},L=$ab$,x=2cm,y=0cm]{v01}
            \Vertex[style={shape=circle},L=$ba$,x=-2cm,y=0cm]{v10}
            \Vertex[style={shape=circle},L=$ca$,x=1cm,y=-1.5cm]{v20}
            \Vertex[style={shape=circle},L=$ac$,x=-1cm,y=-1.5cm]{v02}
            \Edge[style={post, bend left},label=$a$](v00)(v01)
            \Edge[style={post, bend left},label=$b$](v10)(v00)
            \Edge[style={post, bend left},label=$b$](v10)(v01)
            \Edge[style={post, bend right},label=$b$](v10)(v02)
            \Edge[style={post, bend left},label=$a$](v01)(v10)
            \Edge[style={post, bend right},label=$c$](v20)(v01)
            \Edge[style={post, bend right},label=$a$](v02)(v20)
    \end{tikzpicture}
    \end{center}
    The Rauzy group of $ba$ is $\Gr(ba) = \gen{baa,ba,baca}$. 
\end{example}

\begin{proposition}\label{P:equivalence with Rauzy groups}
    Let $X$ be a minimal shift space over $\cA$ and let $\varphi\from F_\cA\to G$ be any group morphism. Then, $X$ has eventually $\varphi$-stable return groups if and only if there exists $K \in\nn$ such that, for all $w \in \lang_{K}(X)$ and all $u \in w\cA^* \cap \lang(X)$, we have $\varphi(\Gr(u)) = \varphi(\Gr(w))$. Moreover, we then have $\varphi(\Gr(u)) = \varphi\gen{\ret_X(u)}$ for all $u \in \lang_{\geq K}(X)$.

    In particular, $X$ has $\varphi$-stable return groups if and only if $\varphi(\Gr(u)) = \varphi(F_\cA)$ for all $u \in \lang(X)$.
\end{proposition}

\begin{proof}
    The proof relies on the following observations: 
    on the one hand, for any $u \in \lang(X)$, every return word to $u$ is the label of a path based on $u$, so $\gen{\ret_X(u)} \leq \Gr(u)$. On the other hand, if $u \in w\cA^* \cap \lang(X)$ is of length at least $K_w = \max\{|vw| \mid v \in \ret_X(w)\}$, then $\Gr(u) \leq \gen{\ret_X(w)}$ by the proof of \cite[Theorem 4.7]{Berthe2015}.
    Consequently, for any $w \in \lang(X)$ and $u \in w\cA^* \cap \lang_{\geq K_w}(X)$,
    \begin{enumerate}
        \item if $\varphi\gen{\ret_X(u)} = \varphi\gen{\ret_X(w)}$, then $\varphi(\Gr(u)) \leq \varphi\gen{\ret_X(w)} = \varphi\gen{\ret_X(u)} \leq \varphi(\Gr(u))$ so $\varphi(\Gr(u)) = \varphi\gen{\ret_X(u)}$;\label{item:equal return group implies}
        \item if $\varphi(\Gr(u)) = \varphi(\Gr(w))$, then $\varphi\gen{\ret_X(w)} \leq \varphi(\Gr(w)) = \varphi(\Gr(u)) \leq \varphi\gen{\ret_X(w)}$ so $\varphi(\Gr(w)) = \varphi\gen{\ret_X(w)}$.\label{item:equal Rauzy group implies}
    \end{enumerate}
    Assume that $X$ has eventually $\varphi$-stable return groups of threshold $M$, and set $K = \max\{K_v \mid v \in \lang_M(X)\}$. Since $K \geq M$, to conclude that $\varphi(\Gr(u)) = \varphi(\Gr(w))$ whenever $w$ is the length-$K$ prefix of $u$, it is sufficient to show that $\varphi(\Gr(u)) = \varphi\gen{\ret_X(u)}$ for all $u \in \lang_{\geq K}(X)$.
    Applying \cref{item:equal return group implies} with $w$ the length-$M$ prefix of $u \in \lang_{\geq K}(X)$, we have $\varphi\gen{\ret_X(u)} = \varphi\gen{\ret_X(w)}$ so $\varphi(\Gr(u)) = \varphi\gen{\ret_X(u)}$. This shows one implication.
    
    For the converse, assume that $\varphi(\Gr(u)) = \varphi(\Gr(w))$ whenever $w$ is the length-$K$ prefix of $u$. To conclude that $X$ has eventually $\varphi$-stable return groups (of threshold at most $K$), it suffices once again to prove that $\varphi(\Gr(w)) = \varphi\gen{\ret_X(w)}$ for all $w \in \lang_{\geq K}(X)$. This is a consequence of \cref{item:equal Rauzy group implies} since, taking any $u \in w\cA^* \cap \lang_{\geq K_w}(X)$, we have $\varphi(\Gr(u)) = \varphi(\Gr(w))$ so $\varphi(\Gr(w)) = \varphi\gen{\ret_X(w)}$.

    The characterization of $\varphi$-stability then directly follows from the fact that $\varphi(\Gr(u)) = \varphi\gen{\ret_X(u)}$ for all long enough $u$ and that $\Gr(u) \leq \Gr(w)$ and $\gen{\ret_X(u)} \leq \gen{\ret_X(w)}$ whenever $w$ is a prefix of $u$.
\end{proof}

\subsection{Behavior of known families}\label{S:examples}

Return groups and their stabilization have already been studied, sometimes implicitly, for some well-known classes of shift spaces. As a consequence, we obtain some simple examples of shift spaces with (eventually) stable return groups.

\subsubsection{Eventually stable return groups}
It is proved in \cite{Berthe2015} that minimal dendric shift spaces --- or more generally, connected shift spaces --- have stable return groups. The second author extended this result by considering a weaker version of connectedness called ``suffix-connectedness'' in~\cite{Goulet-Ouellet2022}. Let us recall the definitions.

Let $X$ be a shift space and let $u \in \lang(X)$. We let $\leftExt{d}(u)$ (resp., $\rightExt{d}(u)$) denote the set of length-$d$ words $v$ such that $vu \in \lang(X)$ (resp., $uv \in \lang(X)$). The \emph{order-$d$ extension graph} of $u$, denoted $\ExtGraph{d}(u)$, is then defined as the bipartite graph with set of vertices $\leftExt{d}(u) \sqcup \rightExt{d}(u)$ and set of edges given by $\{ (v, w) \in \leftExt{d}(u) \times \rightExt{d}(u) \mid vuw \in \lang(X)\}$.

\begin{example}\label{eg:ExtGraph}
    Let us compute the order-2 extension graph of $b$ in the shift space $X$ generated by the Tribonacci substitution defined in \cref{eg:Tribo}. Once again, the length-$3$ factors are $aab,aba,aca,baa,bab,bac$, and $cab$.
    In particular, we deduce that $\leftExt{2}(b) = \{aa,ba,ca\}$ and $\rightExt{2}(b) = \{aa,ab,ac\}$. Exploiting the specific substitution generating $X$, we check that $aabac$, $babac$, $cabaa$, $cabab$, and $cabac$ are in $\lang(X)$ and that $aabaa$, $aabab$, $babaa$, and $babab$ are not. Therefore, the order-2 extension graph of $b$ is the graph represented below.
    
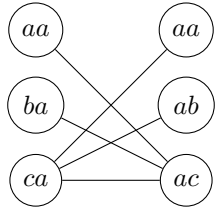
\begin{center}
    \begin{tikzpicture}
    \node[draw,circle] (raa) at (1,1) {$aa$};
    \node[draw,circle] (rab) at (1, 0) {$ab$};
    \node[draw,circle] (rac) at (1, -1) {$ac$};
    
    \node[draw,circle] (laa) at (-1,1) {$aa$};
    \node[draw,circle] (lba) at (-1, 0) {$ba$};
    \node[draw,circle] (lca) at (-1, -1) {$ca$};
    
    \draw[-] (laa) to node{} (rac);
    \draw[-] (lba) to node{} (rac);
    \draw[-] (lca) to node{} (raa);
    \draw[-] (lca) to node{} (rab);
    \draw[-] (lca) to node{} (rac);
    \end{tikzpicture}
    \end{center}
\end{example}

\begin{definition}
    A word $u \in \lang(X)$ is \emph{dendric} if the order-1 extension graph $\ExtGraph{1}(u)$ is a tree (i.e., both connected and acyclic). 
    We then say that a shift space $X$ is {\em eventually dendric} if there exists $M \in \nn$ such that all words $u \in \lang_{\geq M}(X)$ are dendric. The smallest $M$ satisfying this property is called the {\em threshold for dendricity}. If $M = 0$, we say that $X$ is {\em dendric}.
\end{definition}

\begin{definition}    
    A word $u \in \lang(X)$ is \emph{suffix-connected} if there exists $1 \leq d \le |u| + 1$ such that the elements of $\{au' \mid a \in \leftExt{1}(u)\}$ lie in a single connected component of $\ExtGraph{d}(u'')$, where $u = u'u''$ and $|u'| = d-1$. 
    We then say that a shift space $X$ is {\em eventually suffix-connected} if there exists $M \in \nn$ such that all words $u \in \lang_{\geq M}(X)$ are suffix-connected.
    The smallest $M$ satisfying this property is called the {\em threshold for suffix-connectedness}. If $M = 0$, we say that $X$ is {\em suffix-connected}. 
\end{definition}

Obviously, every eventually dendric shift space is also eventually suffix-connected and the threshold for suffix-connectedness cannot exceed the one for dendricity. Observe that, unlike the original definition in~\cite{Goulet-Ouellet2022}, for a shift space to be suffix-connected, we also ask that the empty word is suffix-connected, leading to the following result.

\begin{proposition}[{\cite[Corollary 1.2]{Goulet-Ouellet2022}}]\label{P:suffix-connected}
    Every minimal suffix-connected shift space has stable return groups.
\end{proposition}

The above proposition can be extended to the case where the empty word is not suffix-connected, as shown in \cite{Goulet-Ouellet2022}. More generally, the techniques used in~\cite{Berthe2015} and~\cite{Goulet-Ouellet2022} can be applied to the class of minimal eventually suffix-connected shift spaces, leading to the following natural extension of \cref{P:suffix-connected}.

\begin{proposition}\label{P:eventually suffix-connected are stable}
    Every minimal eventually suffix-connected shift space $X$ has eventually stable return groups.
    Furthermore, the threshold for the stability of return groups does not exceed the threshold for suffix-connectedness.
\end{proposition}

\begin{proof}
    Let $X$ be a minimal eventually suffix-connected shift space of threshold $M$. By the proof of~\cite[Proposition~ 5.3]{Goulet-Ouellet2022}, for every $w \in \lang_M(X)$ and $u \in w\cA^* \cap \lang(X)$, we have $\Gr(u) = \Gr(w)$. By \cref{P:equivalence with Rauzy groups}, $X$ then has eventually stable return groups of threshold at most $M$.
\end{proof}

This proposition sheds some light on the behavior of return groups in eventually dendric shift spaces. It also gives a partial answer to~\cite[Question~10.2]{Goulet-Ouellet2022}. Indeed, as the Rauzy groups $\Gr(u)$ and $\Gr(v)$ are conjugate when $|u| = |v|$,  we established that, whenever $X$ is eventually suffix-connected, all but finitely many return groups of $X$ lie in the same conjugacy class. We cannot precisely compute their rank however.

\subsubsection{Unstable return groups}

On the other hand, Berthé and the second author recently showed that the Thue--Morse shift space does not have eventually stable return groups~\cite{Berthe2023}. This result can be extended to any minimal automatic shift space using a result by Krawczyk and Müllner. Let us first recall the definition of automatic shift spaces.

\begin{definition}
    A shift space $X$ is \emph{$k$-automatic} for $k \geq 2$ if it is generated by $(\sigma, \tau)$ where $\sigma\from \cB^* \to \cB^*$ satisfies $|\sigma(b)| = k$ for all $b \in \cB$, and $\tau\from \cB^* \to \cA^*$ is letter-to-letter. 
\end{definition}

\begin{theorem}[{\cite[Theorem 3.2.3]{these_Krawczyk}}]\label{thm:return_words_of_length_kn}
    If $X$ is an aperiodic minimal $k$-automatic shift space then, for each $n \in\nn$, there is a constant $K_n$ such that, for every $u \in \lang_{\geq K_n}(X)$, any return word to $u$ has length multiple of $k^n$.    
\end{theorem}

\begin{corollary}\label{c:automatic}
    If $X$ is an aperiodic minimal $k$-automatic shift space over $\cA$, then $X$ does not have eventually $\varphi$-stable return groups, where $\varphi\from F_\cA \to \Z$ is such that $\varphi(u) = |u|$ for all $u \in \cA^*$. In particular, $X$ does not have eventually stable return groups.
\end{corollary}

\begin{proof}
    It is enough to prove that for every $w \in \lang(X)$, there exists $u \in w\cA^* \cap \lang(X)$ such that $\varphi\gen{\ret_X(u)} \ne \varphi\gen{\ret_X(w)}$.
    Let us fix such a word $w$.
    There is a positive integer $n$ such that $w$ has a return word of length not multiple of $k^n$. 
    By~\cref{thm:return_words_of_length_kn}, there exists $u$ having $w$ as a prefix and such that all return words to $u$ have length multiple of $k^n$.
    In other words, $\varphi\gen{\ret_X(u)} \subseteq k^n\mathbb{Z}$ while $\varphi\gen{\ret_X(w)} \not\subseteq k^n\mathbb{Z}$, which implies that $\varphi\gen{\ret_X(u)} \ne \varphi\gen{\ret_X(w)}$.
    This shows that the return groups are not eventually $\varphi$-stable. In particular, $X$ does not have eventually stable return groups by \cref{P:free group universal}.
\end{proof}

Combined with \cref{P:eventually suffix-connected are stable}, this implies the following result.

\begin{corollary}
    A minimal eventually suffix-connected shift space that is $k$-automatic is periodic.
\end{corollary}

\subsection{Some equivalent conditions}\label{S:subgroup separable}

When considering stability conditions for a minimal shift space, it is easy to see that it suffices to check that infinitely many words have ``full'' return groups, as detailed below.

\begin{proposition}\label{prop:finitely-many}
    Let $X$ be a minimal shift space over $\cA$ and let $\varphi$ be a group morphism defined on $F_\cA$. The following assertions are equivalent:
    \begin{enumerate}
        \item the shift space $X$ has $\varphi$-stable return groups;
        \item $\varphi\gen{\ret_X(u)} = \varphi(F_\cA)$ for infinitely many $u\in\lang(X)$.
    \end{enumerate}
\end{proposition}

\begin{proof}
    Clearly, if $X$ has $\varphi$-stable return groups, then $\varphi\gen{\ret_X(u)} = \varphi\gen{\ret_X(\emptyw)} = \varphi(F_\cA)$ holds for infinitely many $u \in \lang(X)$. Conversely, let $w \in \lang(X)$. By minimality and by assumption, there exists $u \in \lang(X)$ having $w$ as a factor and such that $\varphi\gen{\ret_X(u)} = \varphi(F_\cA)$. \cref{lem:return groups of factors} then implies that $\varphi(F_\cA)$ is, up to conjugacy by an element of $\varphi(F_\cA)$, a subgroup of $\varphi\gen{\ret_X(w)}$, which in turns implies that $\varphi\gen{\ret_X(w)} = \varphi(F_\cA)$. Therefore, $\varphi\gen{\ret_X(w)} = \varphi(F_\cA) = \varphi\gen{\ret_X(\emptyw)}$ for all $w \in \lang(X)$, proving $\varphi$-stability.
\end{proof}

Building on similar ideas, we provide in \cref{thm:main} different equivalent definitions for eventual $\varphi$-stability. In order to do so, we ask for an additional algebraic condition known as \emph{subgroup separability}. This condition has a rich history in group theory and, as we shall see, is satisfied by the morphisms classically considered in the context of return words. The definition of subgroup separability will be stated in terms of the following notion.

\begin{definition}
    Let $G$ be a group and let $H \leq G$. We say that $G$ is \emph{$H$-separable} if, for every $g\notin H$, there exists a normal subgroup $N\trianglelefteq G$ of finite index such that $gN$ is not in the subgroup $H/N$ of $G/N$.
\end{definition}

In the terminology of Malcev~\cite{Malcev1983}, who was among the firsts to study such properties, we would alternatively say that the subgroup $H$ is a \emph{finitely separable subgroup} of $G$.
The next lemma is a property which appeared in a paper of Raptis and Varsos (\cite[Proposition~1]{Raptis1996}). We include a proof for the sake of completeness.
\begin{lemma}\label{lem:LERF}
    Let $G$ be a group and let $H\leq G$. If $G$ is $H$-separable, then $H$ does not strictly contain any of its conjugates.
\end{lemma}

\begin{proof}
    Suppose by contradiction that the subgroup $g^{-1}Hg$ is strictly contained in $H$ for some $g\in G$. Let $h\in H\setminus g^{-1}Hg$, we have $ghg^{-1} \notin H$. Therefore, there exists a normal subgroup $N \trianglelefteq G$ of finite index such that $ghg^{-1}N \notin H/N$. In particular, $hN \notin g^{-1}Hg/N$.  
    Next, consider the two subgroups $H/N$ and $g^{-1}Hg/N$ of $G/N$. Those are conjugate subgroups in a finite group, thus they have the same cardinality. Moreover $g^{-1}Hg\leq H$, thus $g^{-1}Hg/N = H/N$. We deduce that
    \begin{equation*}
        hN \in H/N = g^{-1}Hg/N,
    \end{equation*}
    contradicting the choice of $N$.
\end{proof}

\begin{definition}
    A group $G$ is \emph{subgroup separable} if it is $H$-separable for every finitely generated subgroup $H\leq G$. 
\end{definition}

Subgroup separable groups are sometimes called \emph{locally extended residually finite} (LERF), a terminology which can be traced back to a paper of Burns~\cite{Burns1971} who attributed it to Meskin. This alternative terminology is meant to highlight the fact that the property is a generalization of \emph{residual finiteness} (RF), which may be defined as $\{1\}$-separability.

Many important families of groups are known to be subgroup separable. A first example is the family of finite groups. Indeed, in a finite group $G$, the trivial subgroup $\{1\}$ has finite index, and if $H \leq G$ is a proper subgroup with $g\notin H$, then $g \notin H = H/\{1\}$. A classical result of M.~Hall, which is particularly important in the context of this paper, states that all free groups are subgroup separable (\cite[Theorem~5.1]{Hall1949}). Another result, due to Malcev, states that \emph{virtually polycyclic groups} (i.e., finite extensions of groups with cyclic subnormal series) have finitely separable subgroups, and thus they are subgroup separable~\cite{Malcev1983}. This includes all finitely generated Abelian groups; see \cite[Lemma~4.2]{Steinberg1999} for a direct proof of this. On the other hand, a simple example of a finitely generated group which is not subgroup separable is the direct product of two free groups of rank 2~\cite{Allenby1973}.

We now provide equivalent points of view for eventual stability.

\begin{theorem}\label{thm:main}
    Let $X$ be a minimal shift space over $\cA$ and let $\varphi\from F_\cA\to G$ be a group morphism onto a group $G$. Consider the following assertions:
    \begin{enumerate}
        \item\label{item:stable return groups} $X$ has eventually $\varphi$-stable return groups;
        \item\label{item:finite return groups} $\{\varphi\gen{\ret_X(u)} \mid u \in \lang(X)\}$ is finite;
        \item\label{item:all long enough are conjugate theorem} the subgroups $\varphi\gen{\ret_X(u)}$ are conjugate for all but finitely many $u \in \lang(X)$;
        \item\label{item:infinitely many conjugate theorem} the subgroups $\varphi\gen{\ret_X(u)}$ are conjugate for infinitely many $u \in \lang(X)$.
    \end{enumerate}
    Then: \ref{item:stable return groups} implies \ref{item:finite return groups} and \ref{item:all long enough are conjugate theorem}; \ref{item:finite return groups} and \ref{item:all long enough are conjugate theorem} each imply \ref{item:infinitely many conjugate theorem}. Moreover, all the assertions are equivalent when $G$ is subgroup separable.
\end{theorem}

\begin{proof}
    \cref{item:stable return groups} implies \cref{item:finite return groups} since it suffices to consider return groups associated with the elements of $\lang_{\leq M}(X)$ for some $M$, which is finite. It also implies \cref{item:all long enough are conjugate theorem} by \cref{P:equivalence with Rauzy groups}. Indeed, there exists $M$ such that, for all $u \in \lang_{\geq M}(X)$, the subgroup $\gen{\ret_X(u)}$ is a Rauzy group for a length-$M$ word. The conclusion then follows since, when $X$ is uniformly recurrent, the Rauzy graphs are strongly connected so Rauzy groups associated with words of the same length are conjugate.
    The fact that \cref{item:finite return groups} implies \cref{item:infinitely many conjugate theorem} is a consequence of the pigeonhole principle, and the implication from \cref{item:all long enough are conjugate theorem} to \cref{item:infinitely many conjugate theorem} is direct.

    Let us now prove that, when $G$ is subgroup separable, \cref{item:infinitely many conjugate theorem} implies \cref{item:stable return groups}. Let $W$ be an infinite subset of $\lang(X)$ such that, for all $v \in W$, the subgroup $\varphi\gen{\ret_X(v)}$ is conjugate to some given $H \leq G$. Fix $v \in W$. As $X$ is minimal, there exists $M$ such that all words $u \in \lang_{\geq M}(X)$ have $v$ as a factor. Let $u \in \lang_{\geq M}(X)$ and let $w$ be its length-$M$ prefix. By Lemma~\ref{lem:return groups of factors} and by definition of $v$, there exists a conjugate $H_1$ of $H$ such that
    \[
        \varphi\gen{\ret_X(u)} \leq \varphi\gen{\ret_X(w)} \leq H_1.
    \]
    Similarly, as $X$ is minimal and $W$ infinite, there exists $t \in W$ such that $u$ is a factor of $t$. Therefore, there exists a conjugate $H_2$ of $H$ such that $H_2$ is a subgroup of $\varphi\gen{\ret_X(u)}$. Since $G$ is subgroup separable and return groups are finitely generated, this implies that $H_1 = H_2$ by \cref{lem:LERF}, so $\varphi\gen{\ret_X(u)} = \varphi\gen{\ret_X(w)}$. As this is true for any $w \in \lang_M(X)$ and $u \in w\cA^* \cap \lang(X)$, this ends the proof.    
\end{proof}

We do not know however if some or all of these assertions are still equivalent when $G$ is not subgroup separable.

We deduce the following dichotomy as a direct consequence of \cref{lem:return groups of factors} and of \cref{thm:main}.
\begin{corollary}\label{cor:dualite}
    Let $X$ be a minimal shift space over $\cA$ and let $\varphi\from F_\cA\to G$ be a morphism onto a subgroup separable group $G$. Either $X$ has eventually $\varphi$-stable return groups, or there exists a sequence $\left(u^{(n)}\right)_{n\in\N} \in \lang(X)^\N$ such that $\varphi\gen{\ret_X\left(u^{(n+1)}\right)} < \varphi\gen{\ret_X\left(u^{(n)}\right)}$ for all $n$.
\end{corollary}

\section{Decidability in the case of substitutive shift spaces}\label{S:decidability}

Based on \cref{prop:finitely-many} and \cref{thm:main}, (eventual) stability of return groups reduces to the study of return groups for well-chosen infinite families of words. This is especially useful in the case of a minimal substitutive shift space $X$ where, following the work of Durand on derivation, one can find a sequence of words $(u^{(n)})_{n \in \N}$, a finite alphabet $\cC$, an endomorphism $\alpha$ of $F_\cC$, and a free group morphism $\psi$ such that, for each $n$, we have $\gen{\ret_X(u^{(n)})} = \psi(\alpha^n(F_\cC))$ (see~\cref{prop:durand}). Moreover, the above morphisms $\alpha$ and $\psi$ are computable, leading to decidability results on eventual stability, both in finite groups and in the free group over the alphabet of $X$.

\subsection{Derived sequences}

Assume that $X$ is a minimal shift space over $\cA$, and consider a word $u \in \lang(X)$. 
Let $\theta_u\from \cB^* \to \cA^*$ define a bijection between $\cB$ and $\ret_X(u)$. There is a shift space $Y$ over $\cB$ such that $\lang(Y) = \theta_u^{-1}(\lang(X))$. This shift space is minimal and is unique up to a renaming of the letters (because $\theta_u$ and $\cB$ are not uniquely defined). It is called the {\em derived shift space of $X$} with respect to $u$, and is denoted $D_u(X)$. The substitution $\theta_u$ is called a {\em derivating substitution} (for $u$).

One can also define derivation of a right-infinite word $x\in\cA^\nn$ with respect to one of its prefixes $u$. In this case, we have a canonical choice of $\theta_u$. More precisely, we naturally set
\begin{equation*}
    \ret_x(u) = \left\{ r\in \cA^+ \mid ru\in \lang(x)\cap (u\cA^*\setminus \cA^+u\cA^+)\right\}
\end{equation*}
and $\cA_u = \{1,\dots,\card\ret_x(u)\}$. Then we define $\theta_{x,u}\from \cA_u^* \to \cA^*$ such that $\theta_{x,u}(\cA_u) = \ret_x(u)$ and for every $i \in \cA_u$, if $k$ is the smallest occurrence of $\theta_{x,u}(i)u$ in $x$, then $x_{[0,k)}$ belongs to $\{\theta_{x,u}(j) \mid 1 \leq j < i\}^*$. In particular, this implies that $\theta_{x,u}(1)u$ is a prefix of $x$.
The {\em derived sequence of $x$} with respect to $u$, denoted $D_u(x)$, is then the unique $z \in \cA_u^\N$ such that $\theta_{x,u}(z) = x$. It satisfies (for the choice $\theta_u = \theta_{x,u}$)
\begin{equation*}
    D_u(X) = \left\{z \in \cA_u^\mathbb{Z} \mid \lang(z) \subseteq \lang(D_u(x))\right\}.  
\end{equation*}

\begin{lemma}[{\cite[Proposition 2.6]{Durand1998}}]
\label{l:durand derivation en chaine}
    Let $X$ be a minimal shift space. If $\theta_u$ is a derivating substitution for $u \in \lang(X)$, then for every $w \in \lang(D_u(X))$, we have
    \[
        \theta_u(\ret_{D_u(X)}(w)) = \ret_X(\theta_u(w)u).
    \]
    Furthermore, if $x = y_{[0,\infty)}$ for some $y \in X$, $u$ is a prefix of $x$, and $w$ is a prefix of $D_u(x)$, then $\theta_u(w)u$ is a prefix of $x$, $D_{\theta_u(w)u}(x) = D_w(D_u(x))$, and $\theta_{x,\theta_u(w)u} = \theta_{x,u}\theta_{D_u(x),w}$.
\end{lemma}

The two following results are reformulations of Durand's work in~\cite{Durand1998,Durand2012}. As \cref{prop:durand} was not explicitly stated, we give a short proof.

\begin{theorem}\label{thm:durand_caclculable}
    Let $\sigma\from \cB^* \to \cB^*$ be a primitive substitution with a fixed point $y\in \cB^{\N}$, let $\tau\from \cB^* \to \cA^*$ be a letter-to-letter substitution, and let $x = \tau(y)$. 
    \begin{enumerate}
        \item For every prefix $u$ of $x$, the substitution $\theta_{x,u}$ is computable and there exist some computable substitutions $\sigma_{u}\from \cC^* \to \cC^*$ and $\tau_u\from \cC^* \to \cA_u^*$ such that $\sigma_u$ is primitive, $\tau_u$ is letter-to-letter, and $\mathcal{D}_u(x) = \tau_{u}(z)$ where $z\in\cA^\nn$ is the fixed point of $\sigma_u$ starting with 1.
        \item There is a computable constant $K$ such that the set 
            \begin{equation*}
                \card\{(\sigma_{u},\tau_{u}) \mid u \text{ prefix of } x\} \leq K.
            \end{equation*}
        In particular, the number of derived sequences of $x$ (with respect to its prefixes) is at most $K$. 
    \end{enumerate}
\end{theorem}

\begin{proposition}\label{prop:durand}
    Let $X$ be a minimal substitutive shift space generated by $(\sigma,\tau)$, where $\sigma\from \cB^* \to \cB^*$ is primitive and $\tau\from \cB^* \to \cA^*$ is letter-to-letter. There exist two computable substitutions $\alpha\from \cC^* \to \cC^*$ and $\psi\from \cC^* \to \cA^*$, and a sequence $(u^{(n)})_{n \in \nn} \in \lang(X)^\nn$ such that $u^{(n)}$ is a proper prefix of $u^{(n+1)}$ and $\ret_X(u^{(n)}) = \psi \alpha^n(\cC)$ for all $n \in\nn$.
\end{proposition}

\begin{proof}
    We set $x = \tau(y)$, where $y$ is a fixed point of $\sigma$.
    Using \cref{l:durand derivation en chaine}, we inductively define a sequence $(v^{(n)})_{n\in\nn}$ of prefixes of $x$ by 
    \[
        v^{(0)} = \emptyw \quad \text{and} \quad  v^{(n+1)} = \theta_{x,v^{(n)}}(1)v^{(n)} \text{ for } n\in\nn.
    \]
    By \cref{l:durand derivation en chaine} once again, the sequence of infinite words defined by $x^{(n)} = D_{v^{(n)}}(x)$ satisfies $x^{(n+1)} = D_{1}(x^{(n)})$ for all $n \in\nn$. Moreover, if we set $\theta_0 = \theta_{x,\emptyw}$ and $\theta_{n+1} = \theta_{x^{(n)},1}$ for all $n$, then $\theta_{x,v^{(n)}} = \theta_0 \dots \theta_n$.
    
    Using~\cref{thm:durand_caclculable} and its notations, there exist computable $i<j$ such that $(\sigma_{v^{(i)}},\tau_{v^{(i)}}) = (\sigma_{v^{(j)}},\tau_{v^{(j)}})$, hence $x^{(i)} = x^{(j)}$.
    This inductively implies that, for all $m \geq i+1$, $x^{(m+j-i)} = x^{(m)}$ and $\theta_{m+j-i} = \theta_{m}$.
    It then suffices to define the sequence $(u^{(n)})_{n \in \nn}$ by $u^{(n)} = v^{(i+n(j-i))}$ for every $n$ and to consider $\psi = \theta_0 \cdots \theta_i$, $\alpha = \theta_{i+1} \cdots \theta_j$, and $\cC = \left\{1,\dots,\card\ret_x(v^{(j)})\right\}$.
\end{proof}

Observe that in the case of a purely substitutive shift space, i.e., $\tau = \id$, the morphisms $\tau_{v^{(i)}}$ in the above proof are also identity morphisms. This is the case for the next example.

\begin{example}\label{ex:unimodular-1}
    Let us illustrate the above construction with the primitive substitution $\sigma\from a \mapsto aab, b \mapsto acb, c \mapsto ba$ and its fixed point 
    \[
        x = aabaabacbaabaabacbaabbaacbaaba\cdots
    \]
    
    We start with $v^{(0)}=\emptyw$. The values of $\theta_0$ and $\sigma_{v^{(0)}}$ are simply
    \begin{align*}
        \theta_0 \from\displaysub{1\mapsto a \\ 2\mapsto b \\ 3\mapsto c}
        \quad \text{and}\quad
        \sigma_{v^{(0)}}\from\displaysub{1 \mapsto 112 \\ 2 \mapsto 132 \\ 3 \mapsto 21}.
    \end{align*}
    In particular $x^{(0)}$ is the image of $x$ under $a\mapsto 1, b\mapsto 2, c\mapsto 3$.

    Next we find $v^{(1)} = \theta_0(1)v^{(0)} =a$. The values of $\sigma_{v^{(1)}}$ and $\theta_1$ may be computed using the procedure described by Durand~\cite{Durand2012}:
    \begin{align*}
        \theta_1 \from \displaysub{1\mapsto 1 \\ 2\mapsto 12 \\ 3\mapsto 132 \\ 4\mapsto 122}
        \quad\text{and}\quad
        \sigma_{v^{(1)}} \from \displaysub{1\mapsto 12 \\ 2\mapsto 123 \\ 3\mapsto 1413 \\ 4\mapsto 1233}.
    \end{align*}

    The next value is $v^{(2)}=\theta_0\theta_1(1)v^{(1)} = aa$. Using the same procedure as before, we find
    \[
        \theta_2 \from \displaysub{ 1\mapsto 12 \\ 2\mapsto 123 \\  3\mapsto 14 \\  4\mapsto 13 \\  5\mapsto 1233 }
            \quad\text{and}\quad
        \sigma_{v^{(2)}} \from \displaysub{ 1\mapsto 12 \\  2\mapsto 1234 \\  3\mapsto 15 \\  4\mapsto 134 \\ 5\mapsto 123434 }.
    \]

    Finally, $v^{(2)}=\theta_0\theta_1\theta_2(1)v^{(2)} = aabaa$ and
    \begin{align*}
        \theta_3 \from \displaysub{ 1\mapsto 12 \\ 2\mapsto 1234 \\ 3\mapsto 15 \\ 4\mapsto 134 \\ 5\mapsto 123434 }
            \quad\text{and}\quad
        \sigma_{v^{(3)}} \from \displaysub{ 1\mapsto 12 \\ 2\mapsto 1234 \\ 3\mapsto 15 \\ 4\mapsto 134 \\ 5\mapsto 123434 }.
    \end{align*}
    
    In particular we see that $\sigma_{v^{(3)}} = \sigma_{v^{(2)}}$. Therefore $x^{(k)}=x^{(k+1)}$ for all $k\geq 2$, and thus $\theta_k = \theta_{k+1}$ for all $k\geq 3$. Consequently, for each $n\in\nn$, the set of return words to the prefix $v^{(2+n)}$ of $x$ is $\psi\alpha^n(\{1,\dots,5\})$, where
    \begin{align*}
        \psi=\theta_0\theta_1\theta_2 \from\displaysub{ 1\mapsto aab \\ 2\mapsto aabacb \\ 3\mapsto aabb \\ 4\mapsto aacb \\ 5\mapsto aabacbacb }
            \quad\text{and}\quad
        \alpha=\theta_3 \from \displaysub{ 1\mapsto 12 \\ 2\mapsto 1234 \\ 3\mapsto 15 \\ 4\mapsto 134 \\ 5\mapsto 123434 }.
    \end{align*}
\end{example}

\subsection{Eventual stability in finite groups}\label{ss:decidability for finite groups}

Observe that all shift spaces have eventually $\varphi$-stable return groups whenever the image of $\varphi$ is a finite group by \cref{thm:main}. We now show that, up to conjugacy, the subgroup over which the images of the return groups stabilize is computable. We deduce in particular that the question of \emph{$\varphi$-stability} in finite groups is decidable. These computability results are motivated by recent work on density of regular languages, where the computation of the stabilizing groups plays a key role~\cite{pre/Berthe2024}.
\begin{proposition}\label{P:H computable}
    Let $X$ be a minimal substitutive shift space generated by $(\sigma,\tau)$, where $\sigma\from \cB^* \to \cB^*$ is primitive and $\tau\from \cB^* \to \cA^*$ is letter-to-letter.
    Let $\varphi\from F_\cA \to G$ be a morphism onto a finite group $G$. One can:
    \begin{enumerate}
        \item compute a subgroup $H$ of $G$ over which the images under $\varphi$ of the return groups stabilize;
        \item decide whether $X$ has $\varphi$-stable return groups.
    \end{enumerate}
\end{proposition}

\begin{proof}
    We prove the first claim.
    By \cref{prop:durand}, there exist computable substitutions $\alpha\from \cC^*\to \cC^*$ and $\psi\from \cC^*\to \cA^*$, and a sequence $(u^{(n)})_{n \in \N} \in \lang(X)^\N$ of distinct words such that $\ret_X(u^{(n)}) = \psi \alpha^n(\cC)$.
    Since $G$ is a finite group (and $\cC$ is finite), there exist positive integers $n,p \leq \card G^{\card\cC}$ such that $\varphi \psi \alpha^n = \varphi \psi \alpha^{n+p}$.
    By induction, this implies that $\varphi \psi \alpha^{n+kp} = \varphi \psi \alpha^n$ for all $k \in \N$. In particular, if $H = \gen{\varphi \psi \alpha^n(\cC)}$ then, for every $k \in \N$, $\varphi\gen{\ret_X(u^{(n+kp)})} = H$. We conclude by \cref{item:infinitely many conjugate theorem} of \cref{thm:main}.
    The second claim directly follows as it suffices to check whether $H = G$.
\end{proof}

\begin{example}\label{ex:unimodular-2}
    Let us illustrate the procedure described in the proof of \cref{P:H computable} with the shift space $X$ generated by the substitution $\sigma\from a\mapsto aac, b\mapsto acb, c\mapsto ba$ from \cref{ex:unimodular-1}. Take the morphism $\varphi\from F_\cA\to S_3$ onto the symmetric group on 3 elements defined by $\varphi\from a\mapsto (1\:2\:3), b\mapsto(1\:2), c\mapsto(1\:2\:3)$ and let $\alpha$ and $\psi$ be as in \cref{ex:unimodular-1}. Then
    \begin{equation*}
        \varphi\psi \from \displaysub{ 1\mapsto (2\:3) \\ 2\mapsto \id \\ 3\mapsto (1\:3\:2) \\ 4\mapsto (1\:2) \\ 5 \mapsto (2\:3) }
        \text{and}\quad
        \varphi\psi\alpha \from \displaysub{ 1\mapsto (2\:3) \\ 2\mapsto \id \\ 3\mapsto \id \\ 4\mapsto \id \\ 5 \mapsto (2\:3) },
    \end{equation*}
    and straightforward computations show that $\varphi\psi\alpha=\varphi\psi\alpha^8$. Hence the images of the return groups of $X$ under $\varphi$ stabilize on conjugates of $\gen{\varphi\psi\alpha(\cC)} = \gen{(2\:3)}$. Note that in this case the computation of $\varphi\psi\alpha$ was already enough to conclude that $X$ does not have $\varphi$-stable return groups.
\end{example}

\subsection{Eventual stability in free groups}

In some particular cases, the substitution $\psi$ of~\cref{prop:durand} is the identity, i.e., we obtain an infinite sequence of words whose return groups are given by $\alpha^n(F_\cC)$. 
More generally, we sometimes have a group morphism $\varphi\from F_\cA \to F_\cA$ and a finitely generated subgroup $H \leq F_\cA$ such that, for all $n \ge 0$, $\varphi^n(H)$ is the return group for some well-chosen $u^{(n)} \in \lang(X)$ (the $u^{(n)}$'s being distinct).
In these cases, we also obtain the decidability of (eventual) stability of return groups.

\begin{lemma}\label{l:decreasing-chain}
    Let $\varphi\from G \to G$ be a group endomorphism and let $H \leq G$. If $\varphi$ is injective over $H$ and $\varphi(H) < H$, then the subgroups $\varphi^n(H)$ form an infinite decreasing chain.
\end{lemma}

\begin{proof}
    Towards a contradition, suppose that $\varphi^n(H) = \varphi^{n+1}(H)$, where $n$ is minimal. Note that $n>0$ as $\varphi(H) < H$. Fix $x\in H$. Then, as $\varphi^n(H)=\varphi^{n+1}(H)$, there exists $y\in H$ such that $\varphi^{n+1}(y) = \varphi^n(x)$. By injectivity, $\varphi^n(y) = \varphi^{n-1}(x)$ as $\varphi^n(y), \varphi^{n-1}(x) \in H$. Since this holds for every $x\in H$, we conclude that $\varphi^{n-1}(H) = \varphi^{n}(H)$, a contradiction to the minimality of $n$.
\end{proof}

The following result relies heavily on algorithmic techniques detailed in \cite{Kapovich2002}, which are based on \emph{Stallings foldings} introduced in Stallings' seminal paper \cite{Stallings1991}.

\begin{lemma}\label{L:decidability for pure morphic}
    Let $X$ be a minimal shift space over $\cA$, let $\varphi\from F_\cA \to F_\cA$ be a group endomorphism, and let $S \subseteq F_\cA$ be a finite subset such that $\gen{\varphi(S)} \leq \gen{S}$. 
    If there exists pairwise distinct words $u^{(n)}\in\lang(X)$ such that $\gen{\ret_X(u^{(n)})} = \gen{\varphi^{n}(S)}$ for all $n \in \N$,
    then one can (from $\varphi$ and $S$):
    \begin{enumerate}
        \item decide whether $X$ has eventually stable return groups;
        \item then compute the subgroup over which the return groups stabilize;
        \item decide whether $X$ has stable return groups.
    \end{enumerate}
\end{lemma}

\begin{proof}
    Before studying the sequence $\gen{\varphi^n(S)}$, let us recall the following generalities on finitely generated subgroups. Let $T$ be a finite subset of $F_\cA$. Then one can compute the rank of $\gen{T}$ using \cite[Proposition~7.1]{Kapovich2002} and \cite[Proposition~8.2]{Kapovich2002} successively. Moreover, $\varphi$ is injective over $\gen{T}$ if and only if $\gen{T}$ and $\varphi\gen{T} = \gen{\varphi(T)}$ have the same rank, and otherwise the rank can only decrease when applying $\varphi$.
    Applied to the sequence $(\gen{\varphi^n(S)})_{n \in \N}$, this leads to the existence of a computable $k \leq \card S$ such that $\sigma$ is injective over $\gen{\varphi^k(S)}$.

    Since $\gen{\varphi(S)} \leq \gen{S}$, we have $\gen{\varphi^{k+1}(S)} \leq \gen{\varphi^k(S)}$.
    To decide whether this is an equality, it then suffices to determine whether $\varphi^k(s) \in \gen{\varphi^{k+1}(S)}$ for all $s \in S$. This can be done using \cite[Proposition~7.1]{Kapovich2002} and \cite[Proposition~7.2]{Kapovich2002}.
    \begin{itemize}
    \item If $\gen{\varphi^{k+1}(S)} = \gen{\varphi^k(S)}$ then, by induction, $\gen{\ret_X(u^{(n)})} = \gen{\varphi^n(S)} = \gen{\varphi^k(S)}$ for all $n \geq k$. By \cref{item:infinitely many conjugate theorem} of \cref{thm:main}, this implies that $X$ has eventually stable return groups, and $\gen{\ret_X(u)}$ is conjugate to $\gen{\varphi^k(S)}$ for any long enough $u \in \lang(X)$.
    \item If $\gen{\varphi^{k+1}(S)} < \gen{\varphi^k(S)}$, then by \cref{l:decreasing-chain}, the subgroups $\gen{\varphi^n(S)}$ form a decreasing chain and $X$ does not have stable return groups by \cref{item:finite return groups} of \cref{thm:main}. 
    \end{itemize}
    This ends the proof of the first claim. Notice that, if $X$ has eventually stable return groups, the return groups then stabilize over $\gen{\varphi^k(S)}$ which is computable. As in \cref{P:H computable}, we can then decide whether $X$ has stable return groups.
\end{proof}

\begin{proposition}\label{P:decidable for derivating and bifix}
    Let $X$ be a shift space generated by a primitive substitution $\sigma$. If $\sigma$ is either a derivating substitution or a bifix substitution, then one can:
    \begin{enumerate}
        \item decide whether $X$ has eventually stable return groups;
        \item compute the subgroup over which the return groups stabilize if they do;
        \item decide whether $X$ has stable return groups.
    \end{enumerate}
\end{proposition}

\begin{proof}
    Assume first that $\sigma\from \cA^* \to \cA^*$ is a derivating substitution. Then there exists a computable $u \in \lang(X)$ such that $\sigma$ is a derivating substitution for $u$, and if $(u^{(n)})_{n\in \N} \in \lang(X)^\N$ satisfies $u_0 = \varepsilon$ and $u^{(n+1)} = \sigma(u^{(n)})u$, we have $\ret_X(u^{(n)}) = \sigma^n(\cA)$ by induction on \cref{l:durand derivation en chaine}. Moreover, we clearly have $\gen{\sigma(\cA)} \leq \gen{\cA}$. We then use \cref{L:decidability for pure morphic} to conclude.

    Next assume that $\sigma$ is bifix. Observe that periodicity of the shift space generated by $\sigma$ is decidable \cite{Pansiot1986,Harju1986} and that the problem becomes trivial in that case: $X$ then has eventually stable return groups, the stabilizing subgroup is generated by a unique word (the period), and in particular, $X$ has stable return group only when it is over a unary alphabet. Hence, we may assume moving forward that $X$ is aperiodic.
    
    By~\cite[Theorem 1]{Berthe2023} and \cite[Proposition 3]{Berthe2023}, there exists a computable positive integer $K$ (which depends on $\sigma$) such that for all $w \in \lang_{\geq K}(X)$, $\ret_X(\sigma(w)) = \sigma(\ret_X(w))$. In particular, for any $u \in \lang_K(X)$, the set $S = \ret_X(u)$ satisfies $\ret_X(\sigma^n(u)) = \sigma^n(S)$ for all $n\in\nn$. Let us show that we can chose $u$ such that it is a prefix of its image under $\sigma^k$ for some $k$. This will then imply that $\gen{\sigma^k(S)} \leq \gen{S}$ by \cref{lem:return groups of factors}. As $\sigma$ is primitive, there exist computable $k \leq \card\cA$ and $a \in \cA$ such that $\sigma^k(a) \in a\cA^+$. We then take $u$ as the length-$K$ prefix of $\lim_n \sigma^{kn}(a)$. Note that $S = \ret_X(u)$ is computable by \cref{thm:durand_caclculable} so we can use \cref{L:decidability for pure morphic} with $\varphi = \sigma^k$ to conclude.
\end{proof}

Before moving on to the next example, let us briefly give further details for the computation of the constant $K$ from \cite[Theorem~1]{Berthe2023} which appears in the above proof. Moving forward we call $K$ (whenever it exists) the \emph{preservation constant}.

Since the preservation constant is related to \emph{synchronization}, we first recall useful terminology taken from \cite{th/Kyriakoglou2019}. Let $\sigma$ be a primitive substitution and let $X$ be the shift space generated by $\sigma$. Given a word $u\in\lang(X)$, an \emph{interpretation} of $u$ is a triple $(p,w,s)$ such that $w\in\lang(X)$ and $\sigma(w) = pus$. We say that an interpretation $(p,w,s)$ \emph{passes by} a factorization $u = u'u''$ if $w = w'w''$ with $\sigma(w') = pu'$ and $\sigma(w'')=u''s$. We say that $u$ is \emph{synchronized} if there is a factorization $u= u'u''$ by which every interpretation of $u$ must pass. Assume that $X$ is aperiodic. Then Mossé's recognizability theorem \cite{Mosse1992} is equivalent to the existence of a \emph{synchronizing constant} $L>0$ such that every word of $\lang_{\geq L}(X)$ is synchronized~\cite[Proposition 3.3.20]{th/Kyriakoglou2019}. If we set
\begin{equation*}
    |\sigma| = \max\{|\sigma(a)| \mid a\in A\},\quad \langle\sigma\rangle =\min\{|\sigma(a)| \mid a\in A\},
\end{equation*}
then, by \cite[Lemma 3.2]{Durand1998}, there is a constant $M$ such that 
\begin{equation*}
    \min\{|r| \mid r\in \ret_X(u)\}\geq |\sigma| \lceil L/\langle\sigma\rangle\rceil
\end{equation*}
for every $u\in \lang_{\geq M}(X)$. Finally, \cite[Theorem~1]{Berthe2023} states that if $\sigma$ is bifix (in addition to being primitive and aperiodic), then it has a preservation constant of $K = \max\{M, 2\lceil L/\langle\sigma\rangle\rceil\}$. Or in other words, it satisfies $\sigma(\ret_X(u)) = \ret_X(\sigma(u))$ for every word $u$ of length at least $K = \max\{M, 2\lceil L/\langle\sigma\rangle\rceil\}$. Note that the constants $L$ and $M$, and thus also the constant $K$, are computable~\cite[Proposition~3]{Berthe2023}.

\begin{example}\label{eg:non-stable}
    Let us continue with the shift space $X$ generated by $\sigma\from a\mapsto aab, b\mapsto acb, c\mapsto ba$ and show that it does not have eventually stable return groups. Observe that $\sigma$ is bifix, primitive and aperiodic. Moreover, $\sigma$ is injective on $F_\cA$, and therefore on any of its subgroups.
    
    Let us compute the constant $K$ as described above. First we observe that all words of length $L=4$ in $\lang(X)$ are synchronized; here are synchronizing factorizations for each of them:
    \begin{gather*}
        cb\cdot ac,\quad
        a\cdot acb,\quad
        ba\cdot ac,\quad
        acb\cdot a,\quad
        cb\cdot aa,\quad 
        aab\cdot a,\quad 
        ab\cdot aa,\quad \\
        b\cdot baa,\quad 
        aab\cdot b,\quad 
        ab\cdot ba,\quad 
        b\cdot aab,\quad 
        b\cdot acb,\quad 
        ab\cdot ac.
    \end{gather*} 
    Hence we find that $\lceil L/\langle\sigma\rangle\rceil=2$, $M=6$ and
    \begin{equation*}
        K = \max(6,2) = 6.
    \end{equation*}
    
    Take for instance $u = aabaab$. The set of return words to $u$ is given by
    \begin{align*}
            \ret_X(u) = \{ & aabaabacb, aabaabbaacb, aabaabacbacb, \\
                           & aabaabacbaabbaacb, aabaabacbaabbaacbaabbaacb\}.
    \end{align*}
    Since $u$ is a prefix of $\sigma(u)$, it follows that $\sigma\gen{\ret_X(u)}$ is a subgroup of $\gen{\ret_X(u)}$. Moreover, one can check that $aabaabbaacb\notin \sigma\gen{\ret_X(u)}$, and thus $\sigma\gen{\ret_X(u)} < \gen{\ret_X(u)}$. Therefore we can apply \cref{L:decidability for pure morphic} with $u^{(n)} = \sigma^n(u)$ to conclude that $X$ does not have eventually stable return groups.
\end{example}

\section{Closure properties}\label{S:closure properties}

This section is devoted to the study of the behavior of (eventual) stability of the return groups under two fundamental operations: derivation and application of a substitution.

\subsection{With respect to derivation}

Derivation being closely related to return words, it is a natural candidate when considering preservation of (eventual) stability of the return groups. 

\begin{proposition}\label{P:stable preserved by derivation}
    Let $X$ be a minimal shift space over $\cA$, let $\varphi\from F_\cA \to G$ be a group morphism, and let $u \in \lang(X)$. If $X$ has eventually $\varphi$-stable return groups of threshold $M$, then $D_u(X)$ has eventually $\varphi\theta_u$-stable return groups of threshold at most $\max\{0,M - |u|\}$, where $\theta_u$ is the associated derivating substitution for $u$.

    In particular, if $X$ has $\varphi$-stable return groups, then $D_u(X)$ has $\varphi\theta_u$-stable return groups.
\end{proposition}

\begin{proof}
    By \cref{l:durand derivation en chaine}, $\varphi\theta_u\gen{\ret_{D_u(X)}(w)} = \varphi\gen{\ret_X(\theta_u(w)u)}$ for all $w \in \lang(D_u(X))$. Therefore, if $K$ satisfies $(|w| \geq K \implies |\theta_u(w)u| \geq M)$, then $D_u(X)$ has eventually $\varphi\theta_u$-stable return groups of threshold at most $K$. The conclusion follows since one can take $K = \max\{0,M - |u|\}$.
\end{proof}

Consequently, if $X$ has eventually $\varphi$-stable return groups of threshold $M$, then $D_u(X)$ has eventually $\psi$-stable return groups of threshold at most $\max\{0,M-|u|\}$ for all $\psi\from F_\cB \to H$ such that $\ker(\varphi\theta_u) \leq \ker(\psi)$. Indeed, $\eta = \psi(\varphi\theta_u)^{-1}\from G \to H$ is then a well-defined group morphism, and $\psi = \eta \varphi\theta_u$. Since $\ker(\theta_u)$ is trivial precisely when $\ret_X(u)$ is a free subset of $F_\cA$, we deduce the following corollary.

\begin{corollary} \label{C:universal+free is stable}
    Let $X$ be a minimal shift space over $\cA$ and let $u \in \lang(X)$ such that $\ret_X(u)$ is a free subset of $F_\cA$. If $X$ has eventually stable return groups of threshold $M$, then $D_u(X)$ has eventually stable return groups of threshold at most $\max\{0,M - |u|\}$.

    In particular, if $X$ has stable return groups, then $D_u(X)$ has stable return groups.
\end{corollary}

Unfortunately, the hypothesis that $\ret_X(u)$ is a free subset of $F_\cA$ is needed: in general, if $X$ has eventually stable return groups, then $D_u(X)$ might not have eventually stable return groups as shown by the following result.

\begin{proposition}\label{p:contre-exemple}
    Let $\sigma\from \{a,b,c,d\}^* \to \{a,b,c,d\}^*$ be the substitution defined by
    \begin{equation*}
        \sigma\from \displaysub{ a\mapsto baa \\ b\mapsto ca \\ c\mapsto bad \\ d\mapsto acd }
    \end{equation*}
    and let $X$ be the shift space generated by $\sigma$. Then
    \begin{enumerate}
        \item $X$ has stable return groups;
            \label{i:contre-exemple-stable}
        \item $D_b(X)$ does not have eventually stable return groups.
            \label{i:contre-exemple-instable}
    \end{enumerate}
\end{proposition}

\begin{proof}
    \ref{i:contre-exemple-stable}. Set $\cA = \{a,b,c,d\}$. As $\sigma$ is bifix, one can use \cref{P:decidable for derivating and bifix} and its proof to determine that $X$ has stable return groups. In this particular case, we can also use a quicker method. Indeed, observe that $\sigma$ is invertible, i.e., $\sigma(F_\cA) = F_\cA$. Therefore it suffices to verify that the sets of return words to words of length $K$ each generate $F_\cA$, where $K$ is the preservation constant detailed in the discussion following the proof of \cref{P:decidable for derivating and bifix}. One can check that, for the morphism $\sigma$ under consideration, $K=10$ (the synchronizing constant is $L=4$, and we find $M=10$ and $K = \max\{10,6\}=10$). Therefore, to show that $X$ has stable return groups, it suffices to show that, for each length-$10$ word, the set of return words generates $F_\cA$. This can be done with explicit computations.
        
    \ref{i:contre-exemple-instable} Using Durand's algorithm, one can show that $D_b(X)$ is generated by the substitution
    \begin{align*}
        \sigma_b \from \displaysub{ 1\mapsto 123334, \\ 2\mapsto 123232533, \\ 3\mapsto 123233, \\ 4\mapsto 12333632734, \\ 5\mapsto 12323232736533, \\ 6\mapsto 12333632736533, \\ 7\mapsto 12323232734.}
    \end{align*}
    Observe that $\sigma_b$ is a derivating substitution (for the letter $1$) so one can apply the procedure of \cref{P:decidable for derivating and bifix} to show that $D_b(X)$ does not have eventually stable return groups. Computations performed in SageMath~\cite{man/Sage2022} with the publicly available \verb|stallings_graphs| package~\cite{man/stallings2019} show that $\sigma_b$ is injective over $H = \gen{\ret_{D_b(X)}(1)} = \gen{\sigma_b(\cB)}$ and that $\sigma_b(H) < H$. Hence, $D_b(X)$ does not have universally stable return groups.
\end{proof}

This negative result does not contradict \cref{C:universal+free is stable} as the set $\ret_{X}(b)$ is not free. For instance $\ret_X(b)$ contains the four elements $baa, baaca, badacd, badacdca$, and they satisfy the relation
\[
    (baa)^{-1} baaca = (badacd)^{-1}badacdca.
\]

\begin{corollary}
    The family of shift spaces with (eventually) stable return groups is not closed under derivation.
\end{corollary}

\subsection{With respect to substitutions}

Another classical operation on shift spaces is to consider the image under a substitution. In the particular case where the substitution is a derivating substitution, this operation is the converse of derivation. \cref{l:durand derivation en chaine} then gives a strong link between the sets of return words. For a general substitution, there is a similar link but the return words in the image are related to the return words to the \emph{set of pre-images}, which might contain more than one word. The study of the behavior of (eventual) stability under application of a substitution therefore relies heavily on the notion of return words to sets of words.

\begin{definition}
    A set $S \subseteq \cA^*$ is a \emph{factor code} if it is non-empty and no word is factor of another.
\end{definition}
If $S \subseteq \lang(X)$ is a factor code, then one naturally defines the return words to $S$ in $X$ as follows:
\[
    \ret_X(S) = \left\{r \in \cA^+ \mid \exists u \in S \text{ st. } ru \in \lang(X) \cap (S\cA^* \setminus \cA^+S\cA^+)\right\}.
\]
Note that, if $S, T$ are two factor codes such that $S \subseteq T$, then $\ret_X(S) \subseteq (\ret_X(T))^+$.

\begin{lemma}\label{L:stable return groups also works for sets}
    Let $X$ be a minimal shift space over $\cA$ and let $\varphi\from F_\cA \to G$ be a group morphism. Then $X$ has eventually $\varphi$-stable return groups of threshold at most $M$ if and only if, for all factor codes $T \subseteq \lang(X)$, we have $\varphi\gen{\ret_X(T)} = \varphi\gen{\ret_X(S)}$, where $S = \{u_{[0,\min\{|u|,M\})} \mid u \in T\}$.
\end{lemma}

\begin{proof}
    One implication directly follows from the fact that singletons are particular factor codes. Let us prove the converse, and let $S,T \subseteq \lang(X)$ be two factor codes such that $S = \{u_{[0,\min\{|u|,M\})} \mid u \in T\}$.

    By definition of $S$ and $T$, $\ret_X(T) \subseteq (\ret_X(S))^*$ so the inclusion $\varphi\gen{\ret_X(T)} \leq \varphi\gen{\ret_X(S)}$ follows. Conversely, take $r \in \ret_X(S)$. There exist $u, u' \in S$ such that $ru' \in \lang(X) \cap u \cA^*$. By definition of $S$, there exist $v, v' \in \cA^*$ such that $uv, u'v' \in T$. Since $X$ is minimal, there exist $w, t \in \lang(X)$ such that $wrtu'v' \in \lang(X) \cap uv \cA^*$ and $u'$ is a prefix of $tu'v'$. Then $wrt \in (\ret_X(T))^*$, $w \in (\ret_X(u))^*$, and $t \in (\ret_X(u'))^*$ so $r \in \gen{\ret_X(u)}\gen{\ret_X(T)}\gen{\ret_X(u')}$. Moreover, as $X$ has eventually $\varphi$-stable return groups of threshold at most $M$, $\varphi\gen{\ret_X(u)} = \varphi\gen{\ret_X(uv)} \leq \varphi\gen{\ret_X(T)}$. Indeed, either $v = \emptyw$ or $|u| = M$. Similarly, $\varphi\gen{\ret_X(u')} \leq \varphi\gen{\ret_X(T)}$. Thus $\varphi(r) \in \varphi\gen{\ret_X(T)}$. As this is true for any $r \in \ret_X(S)$, we conclude that $\varphi\gen{\ret_X(T)} = \varphi\gen{\ret_X(S)}$.
\end{proof}

We can now prove the counterpart of \cref{P:stable preserved by derivation} and \cref{C:universal+free is stable}. For a shift space $X$ over $\cA$ and a substitution $\sigma \from \cA^* \to \cB^*$, we let $\sigma[X]$ denote the shift-closure of $\{\sigma(x) \mid x \in X\}$.

\begin{proposition}\label{P:eventual stability preserved by substitution}
    Let $X$ be a minimal shift space over $\cA$ and let $\sigma\from \cA^* \to \cB^*$ be a substitution. Let $\varphi\from F_\cA \to G$ and $\psi\from F_\cB \to H$ be group morphisms satisfying $\ker(\varphi) \leq \ker(\psi\sigma)$. If $X$ has eventually $\varphi$-stable return groups of threshold $M$, then $\sigma[X]$ has eventually $\psi$-stable return groups of threshold at most $\max\{1,M|\sigma|\}$.
\end{proposition}

\begin{proof}
    First, recall that $|\sigma| = \max\{|\sigma(a)| \mid a \in \cA\}$ and observe that the hypothesis on the kernels is equivalent to $\bar\sigma = \psi\sigma\varphi^{-1}\from \varphi(F_\cA) \to H$ being a well-defined group morphism. This morphism satisfies $\bar\sigma\varphi = \psi\sigma$.
    
    To study $\sigma[X]$, we first factorize the substitution $\sigma$ into $\tau\theta$ as follows.
    Let $\cA = \{a_1, \dots, a_d\}$, $n_i = |\sigma(a_i)|$, and $\cC = \{a_{i,j} \mid 1 \leq i \leq d, 0 \leq j \leq |\sigma(a_i)|-1\}$ where $a_{i,j} \ne a_{i',j'}$ if $i \ne i'$ or $j \ne j'$. We define $\theta\from \cA^*\to \cC^*$ and $\tau\from \cC^*\to \cB^*$ by
    \[
        \theta(a_i) = a_{i,0} \cdots a_{i,n_i-1}\quad\text{and}\quad \tau(a_{i,j}) = \sigma(a_i)_j.
    \]
    They clearly satisfy $\sigma = \tau\theta$.
    
    We consider the intermediary shift space $\theta[X]$ and show that it is eventually $\psi\tau$-stable of threshold at most $\max\{1,M|\sigma|\}$. To do so, we first study the return group of any non-empty $v \in \lang(\theta[X])$. Let $a_{i,j}$ and $a_{k,l}$ respectively be the first and last letters of $v$. Define $p_v = a_{i,0}\cdots a_{i,j-1}$ and $s_v = a_{k,l+1}\cdots a_{k,n_k-1}$. Then $p_vvs_v = \theta(t_v)$ for a unique $t_v\in\lang(X)$. Since $v$ only appears as a factor of $\theta(t_v)$ and $\theta(t_v)$ only appears as the image of $t_v$, we have
    \begin{equation}\label{eq:ret-theta}
        \ret_{\theta[X]}(v) = p_v^{-1}\ret_{\theta[X]}(\theta(t_v))p_v = p_v^{-1}\theta(\ret_X(t_v))p_v.
    \end{equation}
    Now, take $w \in \lang(\theta[X])$ of length $\max\{1,M|\sigma|\}$ and $u \in w\cC^* \cap \lang(\theta[X])$. We prove that $\psi\tau\gen{\ret_{\theta[X]}(w)} = \psi\tau\gen{\ret_{\theta[X]}(u)}$. By \cref{eq:ret-theta}, we have
    \begin{align*}
        \psi\tau\gen{\ret_{\theta[X]}(w)}
        &= \psi\tau(p_w)^{-1}\psi\tau\theta\gen{\ret_X(t_w)}\psi\tau(p_w) \\
        &= \psi\tau(p_w)^{-1}\psi\sigma\gen{\ret_X(t_w)}\psi\tau(p_w) \\
        &= \psi\tau(p_w)^{-1}\bar\sigma\varphi\gen{\ret_X(t_w)}\psi\tau(p_w).
    \end{align*}
    Similarly, $\psi\tau\gen{\ret_{\theta[X]}(u)} = \psi\tau(p_u)^{-1}\bar\sigma\varphi\gen{\ret_X(t_u)}\psi\tau(p_u)$.
    As $p_u = p_w$ by definition of $u$ and $w$, it then suffices to show that $\varphi\gen{\ret_X(t_w)} = \varphi\gen{\ret_X(t_u)}$. By definition of $u$ and $w$ once again, $t_w$ is a prefix of $t_u$ and
    \[
        |t_w| \geq \frac{|\theta(t_w)|}{|\sigma|} = \frac{|p_wwq_w|}{|\sigma|} \geq \frac{|w|}{|\sigma|}
        \geq M.
    \]
    As $X$ is eventually $\varphi$-stable of threshold $M$, this directly implies that $\varphi\gen{\ret_X(t_w)} = \varphi\gen{\ret_X(t_u)}$, and therefore that $\psi\tau\gen{\ret_{\theta[X]}(w)} = \psi\tau\gen{\ret_{\theta[X]}(u)}$. This ends the proof that $\theta[X]$ is eventually $\psi\tau$-stable of threshold $K \leq \max\{1,M|\sigma|\}$.
    
    Let us now show that $\sigma[X]$ is eventually $\psi$-stable of threshold at most $K$. Let $w \in \lang_K(\sigma[X])$ and let $u \in w\cB^* \cap \lang(\sigma[X])$. Since $\tau$ is letter-to-letter, one easily sees that $\tau^{-1}(u)\cap\lang(\theta[X])$ is a factor code and 
    \[
        \ret_{\sigma[X]}(u) = \tau(\ret_{\theta[X]}(\tau^{-1}(u)\cap\lang(\theta[X]))).
    \]
    Set $T = \tau^{-1}(u)\cap\lang(\theta[X])$ and $S = \{v_{[0,K)} \mid v \in T\}$. Note that $S = \tau^{-1}(w)\cap\lang(\theta[X])$ thus we similarly have $\ret_{\sigma[X]}(w) = \tau(\ret_{\theta[X]}(S))$. As $\theta[X]$ is eventually $\psi\tau$-stable of threshold $K$, we obtain by \cref{L:stable return groups also works for sets}
    \[
        \psi\ret_{\sigma[X]}(u) = \psi\tau(\ret_{\theta[X]}(T)) = \psi\tau(\ret_{\theta[X]}(S)) = \psi\ret_{\sigma[X]}(w).
    \]
    This shows that $\sigma[X]$ is eventually $\psi$-stable of threshold at most $K \leq \max\{1,M|\sigma|\}$.
\end{proof}

Observe that the bound $\max\{1,M|\sigma|\}$ can safely be replaced by $M|\sigma|$ except when $M=0$. The fact that $\max\{1,M|\sigma|\}$ is tight when $M=0$ is made clear by the following example.
\begin{example}
    Let $\cA = \{a,b\}$ and let $\sigma\from \cA^*\to\cA^*$ be defined by $\sigma(a) = ab$ and $\sigma(b) = abbb$. For any shift space $X$ over $\cA$, we have $\ret_{\sigma[X]}(a) = \{ab, abbb\}$. In particular, $\sigma[X]$ cannot be stable, even when $X$ is.
\end{example}

The following shows that this phenomenon disapears when the morphism $\sigma$ is surjective between free groups.

\begin{proposition}\label{P:stability preserved by substitution}
    Let $X$ be a minimal shift space over $\cA$, let $\sigma\from \cA^* \to \cB^*$ be a substitution, and let $\varphi\from F_\cA \to G$ and $\psi\from F_\cB \to H$ be group morphisms satisfying $\ker(\varphi) \leq \ker(\psi\sigma)$. If $X$ has $\varphi$-stable return groups and $\sigma\from F_\cA \to F_\cB$ is surjective, then $\sigma[X]$ has $\psi$-stable return groups.
\end{proposition}

\begin{proof}
    For conciseness, we re-use the notations from the proof of \cref{P:eventual stability preserved by substitution}.
    Let $u \in \lang(\sigma[X]) \setminus \{\varepsilon\}$. Then for any $w \in \tau^{-1}(u) \cap \lang(\theta[X])$, we have
    \[
        \psi\gen{\ret_{\sigma[X]}(u)} \geq \psi\tau\gen{\ret_{\theta[X]}(w)} = \psi\tau(p_w)^{-1} \psi\sigma\gen{\ret_X(t_w)} \psi\tau(p_w).
    \]
    Since $X$ has $\varphi$-stable return groups and $\sigma$ is surjective, we have
    \[
        \psi\sigma\gen{\ret_X(t_w)} = \bar\sigma\varphi\gen{\ret_X(t_w)} = \bar\sigma\varphi(F_\cA) = \psi\sigma(F_\cA) = \psi(F_\cB).
    \]
    This shows that $\psi\gen{\ret_{\sigma[X]}(u)}$ contains $\psi(F_\cB)$, up to conjugacy by an element of $\psi(F_\cB)$. Consequently, $\psi\gen{\ret_{\sigma[X]}(u)} = \psi(F_\cB)$, and as it is true for all non-empty $u \in \lang(\sigma[X])$, the shift space $\sigma[X]$ has $\psi$-stable return groups by \cref{prop:finitely-many}.
\end{proof}

Finally, let us highlight two consequences of these propositions.

\begin{corollary}\label{c:closed under substitution}\hfill
    \begin{enumerate}
    \item The family of shift spaces with eventually stable return groups is closed under application of substitutions.
    \item The family of shift spaces with stable return groups is closed under application of substitutions that are surjective when seen as group morphisms.
    \end{enumerate}
\end{corollary}
\begin{proof}
    Indeed, we can apply \cref{P:eventual stability preserved by substitution} and \cref{P:stability preserved by substitution} with $\varphi$ the identity on $F_\cA$, $\psi$ the identity on $F_\cB$, and any substitution $\sigma$ since $\ker(\varphi)$ is trivial.
\end{proof}

\section{Local-global principle for stability}\label{S:local-global}

As stated in \cref{ss:decidability for finite groups}, if the image of $\varphi$ is finite, then any minimal shift space $X$ has eventually $\varphi$-stable return groups. But, as shown in \cref{S:stability and fullness}, not every shift space has eventually stable return groups. This shows that eventual stability of return groups cannot simply be reduced to a ``local'' property. On the other hand, when considering stability, we show that it is sufficient to consider morphisms $\varphi$ with finite images, making this property easier to study. More precisely, we prove the following \emph{local-global} principle for stability.

\begin{proposition}\label{p:local-global}
    Let $X$ be a minimal shift space over $\cA$ and let $\varphi\from F_\cA\to G$ be a group morphism onto a subgroup separable group $G$. For every $u\in\lang(X)$, the following assertions are equivalent:
    \begin{enumerate}
        \item $\varphi\gen{\ret_X(u)} = G$;
        \item $\psi\gen{\ret_X(u)} = H$ for every onto group morphism $\psi\from F_\cA\to H$ where $H$ is finite and $\ker(\varphi)\leq\ker(\psi)$.
    \end{enumerate}
\end{proposition}

To simplify the proof of this result, we first obtain an alternative statement using the following simple observation. 

\begin{remark}\label{R:from morphisms to normal subgroups}
    By the first isomorphism theorem, statements about images of morphisms may be translated in terms of normal subgroups. More precisely, for an onto morphism $\varphi\from G \to H$ and a subset $S \subseteq G$, the property $\varphi(S) = H$ is equivalent to $S\ker(\varphi) = G$. Moreover, $H$ is finite if and only if $\ker(\varphi)$ has finite index.
\end{remark}

In light of this, the statement of \cref{p:local-global} can be rephrased as follows. Let $X$ be a minimal shift space over $\cA$ and let $N$ be a normal subgroup of $F_\cA$ such that $F_\cA/N$ is subgroup separable. Then for every $u \in \lang(X)$, $\gen{\ret_X(u)}N = F_\cA$ if and only if $\gen{\ret_X(u)}N' = F_\cA$ for every normal subgroup $N' \trianglelefteq F_\cA$ of finite index such that $N \leq N'$. The proof relies on a simple lemma.
\begin{lemma}\label{l:local-global}
    Let $G$ be a group and let $H\leq G$ such that $G$ is $H$-separable. If $HN=G$ for all normal subgroups $N\trianglelefteq G$ of finite index, then $H=G$.
\end{lemma}

\begin{proof}
    Let us proceed by contraposition and assume that $H$ is a proper subgroup of $G$. Let $g \in G \setminus H$. Since $G$ is $H$-separable, we may take a normal subgroup $N \trianglelefteq G$ of finite index such that $gN \not \in H/N$. Thus $H/N \neq G/N$, which means that $HN \neq G$.
\end{proof}

\begin{proof}[Proof of \cref{p:local-global}]
    For the forward implication, we assume that $\varphi\gen{\ret_X(u)} = G$, which is equivalent to $\gen{\ret_X(u)}\ker(\varphi) = F_\cA$ by \cref{R:from morphisms to normal subgroups}. Clearly for every morphism $\psi\from F_\cA\to H$ onto a group $H$ such that $\ker(\varphi) \leq \ker(\psi)$, we then have $F_\cA = \gen{\ret_X(u)}\ker(\varphi) \leq \gen{\ret_X(u)}\ker(\psi)$, thus $\psi\gen{\ret_X(u)} = H$.

    For the converse, as $G$ is subgroup separable, it is in particular $\varphi\gen{\ret_X(u)}$-separable.
    By \cref{l:local-global}, to show that $\varphi\gen{\ret_X(u)} = G$, it suffices to prove that $\varphi\gen{\ret_X(u)}N=G$ for every normal subgroup $N\trianglelefteq G$ of finite index. Let $N$ be such a subgroup.
    By the correspondence theorem, there exists a normal subgroup $N'\trianglelefteq F_\cA$ of finite index containing $\ker(\varphi)$ such that $N = \varphi(N')$.
    Let $\psi\from F_\cA\to F_\cA/N'$ denote the canonical morphism. We get $\ker(\varphi) \leq N' = \ker(\psi)$ and $F_\cA/N'$ is finite. Hence, by assumption, $\psi\gen{\ret_X(u)} = F_\cA/N'$, or in other words, $\gen{\ret_X(u)}N' = F_\cA$. Applying $\varphi$, we obtain $\varphi\gen{\ret_X(u)}N=G$. As this is true for every normal subgroup $N\trianglelefteq G$ of finite index, we conclude by \cref{l:local-global}.
\end{proof}

The next corollary highlights the special case of \cref{p:local-global} when $\varphi$ is the identity on $F_\cA$.
\begin{corollary}
    Let $X$ be a minimal shift space over $\cA$ and let $u\in\lang(X)$. The following assertions are equivalent:
    \begin{enumerate}
        \item $\gen{\ret_X(u)} = F_\cA$;
        \item $\psi\gen{\ret_X(u)} = H$ for every onto morphism $\psi\from F_\cA\to H$ where $H$ is a finite group.
    \end{enumerate}
\end{corollary}

Next we turn to the case where $\varphi$ is the Abelianization map, that is,
\begin{equation*}
        \ab\from F_\cA\to\zz^{\card\cA} \quad u \mapsto (|u|_a)_{a \in \cA},
\end{equation*}
where $|u|_a$ is the number of occurrences of $a$ in $u$, with the inverse $a^{-1}$ counting as a negative occurrence. In this context, we can improve slightly the statement of \cref{p:local-global}. To this end, consider the \emph{$k$-Abelianization} morphisms ($k>1$), defined by
\begin{equation*}
    \ab_k\from F_\cA\to(\zz/k\zz)^{\card\cA} \quad u \mapsto (|u|_a + k\Z)_{a \in \cA}.
\end{equation*}

\begin{proposition}\label{P:local-global-abelian}
    Let $X$ be a minimal shift space over $\cA$ and let $u\in\lang(X)$. The following assertions are equivalent:
    \begin{enumerate}
        \item $\ab\gen{\ret_X(u)} = \zz^{\card\cA}$;\label{i:global-abelian}
        \item $\ab_k\gen{\ret_X(u)} = (\zz/k\zz)^{\card\cA}$ for every $k\geq 1$.\label{i:local-abelian}
    \end{enumerate}
\end{proposition}

The proof relies on the following elementary fact, of which we include a proof.
\begin{lemma}\label{l:kernel-ab-k}
    For any $k>1$, the kernel $\ker(\ab_k)$ is the smallest normal subgroup of $F_\cA$ containing all commutators and $k$-th powers.
\end{lemma}

\begin{proof}
    The fact that $\ker(\ab_k)$ is a normal subgroup containing commutators and $k$-th powers is an easy consequence of the definition. The fact that it is the \emph{smallest} such normal subgroup follows from the observation that every element of $g\in F_\cA$ can be written in the form
    \begin{equation*}
        g = a_1^{\delta_1}a_2^{\delta_2}\cdots a_d^{\delta_d}c,
    \end{equation*}
    where $c$ is a product of commutators, $\delta_i\in\zz$ and $\cA = \{a_1,\dots,a_d\}$. It is then clear that, when $g\in\ker(\ab_k)$, all $\delta_i$ must be multiples of $k$, and thus $g$ is a product of $k$-th powers and commutators.
\end{proof}

\begin{proof}[Proof of \cref{P:local-global-abelian}]
    The fact that \ref{i:global-abelian} implies \ref{i:local-abelian} is an immediate consequence of \cref{p:local-global}. 
    
    For the converse, in light of \cref{p:local-global}, we need to show that, for any onto morphism $\psi\from F_\cA\to H$ such that $H$ is finite and $\ker(\ab)\leq\ker(\psi)$, we have $\psi\gen{\ret_X(u)} = H$. It is equivalent to show that $\gen{\ret_X(u)}\ker(\psi) = F_\cA$ by \cref{R:from morphisms to normal subgroups}, or, as we assume that \ref{i:local-abelian} holds, it suffices to prove that $\ker(\psi)$ contains $\ker(\ab_k)$ for some $k\geq 1$.
    
    Let us fix such a morphism $\psi$. As $\ker(\ab)\leq\ker(\psi)$, the latter contains all commutators (and in particular, $H$ is Abelian). Let us show that it contains all $k$-th powers for an appropriate choice of $k \geq 1$. For each $a\in \cA$, let $k_a$ be the order of $\psi(a)$ in $H$, and let $k=\lcm\{k_a \mid a\in \cA\}$. Take $g = a_1^{\delta_1}\cdots a_d^{\delta_d}c \in F_\cA$, where $\delta_i\in\zz$ and $c$ is a product of commutators. Since $H$ is Abelian, we find (in additive notation for $H$) that
    \begin{equation*}
        \psi(g^k) = \sum_{i=1}^d\delta_ik\psi(a_i) = 0,
    \end{equation*} 
    since $k$ is a common multiple of the orders of the $\psi(a_i)$. This shows that $\ker(\psi)$ contains all $k$-th powers in $F_\cA$, thus it contains $\ker(\ab_k)$ by \cref{l:kernel-ab-k}.
\end{proof}

\section{Applications to the welldoc property}
\label{s:welldoc}

It turns out that Abelian stability already appeared in the literature in a disguised form. Namely, in a paper from 2016~\cite{Balkova2016}, Balková et al. studied the \emph{well distributed occurrence} property --- or welldoc for short --- in relation with pseudorandom number generation. The main purpose of this section is to show precisely how welldoc relates to Abelian stability. This also allows us to obtain some sufficient conditions for the welldoc property which generalize results of Balková et al.

The original definition of welldoc is framed within the setting of \emph{one-sided} infinite words, i.e., elements of $\cA^\nn$. It goes as follows: given a one-sided infinite word $y\in \cA^\nn$ and $w \in \lang(y)$, consider the set
\begin{equation*}
    P_y(w)  = \{y_{[0,m)} \mid m \in \N, y_{[m,m+|w|)} = w\}.
\end{equation*}  
Then, in Balková et al.'s terminology, $y \in \cA^\N$ is said to have the welldoc property when, for all $k\geq 1$ and all words $w\in\lang(y)$, 
\begin{equation*}
    \ab_k(P_y(w)) = (\Z/k\Z)^{\card\cA}.
\end{equation*}

Let us extend the notion of welldoc to the setting of general group morphisms and bi-infinite words and shift spaces. To do so, it is useful to set $x_{[i,j)} = x_{[j,i)}^{-1}$ when $j<i$, where the inverse is taken in the free group $F_\cA$. Then, for $x\in \cA^\Z$ and $w\in\lang(x)$, we define the following subset of $F_\cA$:
\begin{equation*}
    P_x(w) = \{ x_{[0,m)} \mid m \in \zz, x_{[m,m+|w|)} = w\}.
\end{equation*}
Note that if $y = x_{[0,\infty)}$, then $P_y(w) = P_x(w)\cap \cA^*$.

\begin{definition}
    Let $\varphi\from F_\cA\to G$ be a group morphism. A word $x\in\cA^\zz \cup \cA^\N$ has \emph{$\varphi$-welldoc} if, for all $w\in\lang(x)$,
    \begin{equation*}
        \varphi(P_x(w)) = \varphi(F_\cA).
    \end{equation*}
    By extension, we say that a shift space $X$ has \emph{$\varphi$-welldoc} when all its elements have it.
\end{definition}

In this terminology, Balková et al.'s definition of welldoc for $y \in \cA^\N$ is equivalent to the conjunction of the $\ab_k$-welldoc properties over all integers $k\geq 1$. 

In the case of finite groups, this two-sided version is coherent with the original one-sided definition. In fact, welldoc is a property of the language rather than the point, as detailed in the following lemma.
In particular, a minimal shift space has $\varphi$-welldoc if and only if \emph{any} of its elements does.

\begin{lemma}\label{L:welldoc one and two sided}
    Let $\varphi\from F_\cA \to G$ be a morphism onto a finite group $G$. Let $x,y\in\cA^\zz\cup\cA^\nn$ be such that $\lang(x)=\lang(y)$. If $x$ has $\varphi$-welldoc, then so does $y$.
    More precisely, for any $w \in \lang(x)$, if $\varphi(P_x(w)) = G$, then $\varphi(P_y(w)) = G$.
\end{lemma}

\begin{proof}
    Let $w \in \lang(x)$. Since $G$ is finite, there exists $M$ such that the elements of $P_x(w)$ of length at most $M$ suffice to obtain $G$. Letting $M'=-M$ if $x\in\cA^\zz$ and $M'=0$ if $x\in\cA^\nn$, this means that for every $g\in G$, there exists $i$ in $[M',M)$ such that $\varphi(x_{[0,i)})=g$ and $x_{[i,i+|w|)}=w$. Since $\lang(y) = \lang(x)$, there exists an index $k$ such that $x_{[M',M+|w|)} = y_{[k+M',k+M+|w|)}$. Therefore, $G \subseteq \varphi(y_{[0,k)})^{-1} \varphi(P_y(w))$, which implies that $\varphi(P_y(w)) = G$.
\end{proof}

Observe that if $w = x_{[0,n)}$, then the positive (resp., negative) elements of $P_x(w)$ are concatenations of return words (resp., inverses of return words) to $w$ in $x$. More generally, we have the inclusion
\begin{equation}\label{eq:inclusion-welldoc}
    P_x(w) \subseteq x_{[0,i)}^{-1}\gen{\ret_X(w)},
\end{equation}
whenever $x_{[i,i+|w|)}=w$.
This hints at a relationship between welldoc and stability. The following result, which was proved in \cite{pre/Berthe2024} using tools from the theory of bifix codes, further confirms this.

\begin{proposition}[{\cite[Theorem 6.1]{pre/Berthe2024}}]\label{p:welldoc berthe et al}
    Let $X$ be a minimal shift space and $\varphi\from F_\cA\to G$ be a morphism onto a finite group $G$. Then $X$ is $\varphi$-stable if and only if, for every $x\in X$ and $n\geq 0$,
    \begin{equation*}
        \varphi(P_{x_{[0,\infty)}}(x_{[0,n)})) = G.
    \end{equation*}
\end{proposition}

Combined with our own results, it can be strengthened as follows.
\begin{theorem}\label{t:welldoc}
    Let $X$ be a minimal shift space over $\cA$ and let $\varphi\from F_\cA\to G$ be a morphism onto a subgroup separable group $G$. The following assertions are equivalent:
    \begin{enumerate}
        \item $X$ has $\varphi$-stable return groups;
        \item $X$ has $\psi$-welldoc for every onto morphism $\psi\from F_\cA\to H$ where $H$ is finite and $\ker(\varphi)\subseteq\ker(\psi)$.
    \end{enumerate}
\end{theorem}

\begin{proof}
    Thanks to the local-global principle (\cref{p:local-global}), it is sufficient to show that, for any morphism $\varphi\from F_\cA \to G$ onto a finite group $G$, $X$ has $\varphi$-stable return groups if and only if it has $\varphi$-welldoc.
    The fact that $\varphi$-welldoc implies $\varphi$-stable return groups follows from \eqref{eq:inclusion-welldoc}.    
    
    Next assume that $X$ has $\varphi$-stable return groups. Fix $w\in\lang(X)$.
    By minimality of $X$, there exists $z \in X$ such that $w = z_{[0,|w|)}$ and $\lang(X) = \lang(z_{[0,\infty)})$. Using \cref{p:welldoc berthe et al}, we have $\varphi(P_{z_{[0,\infty)}}(w)) = G$. 
    Therefore, by \cref{L:welldoc one and two sided}, for any $x \in X$, $\varphi(P_x(w)) = G$. As it is true for any $w \in \lang(X)$ and any $x \in X$, $X$ has $\varphi$-welldoc.
\end{proof}

As in the previous section, we may sharpen the result when $\varphi$ is the Abelianization map by using \cref{P:local-global-abelian} instead of the general the local-global principle (\cref{p:local-global}).
This leads to the following simple characterization of the original welldoc property of Balková et al.
\begin{corollary}\label{t:abelian-welldoc}
    Let $X$ be a minimal shift space over $\cA$. The following assertions are equivalent:
    \begin{enumerate}
        \item $X$ has $\ab$-stable return groups;
        \item $X$ has $\ab_k$-welldoc for all integers $k\geq 1$.
    \end{enumerate}
\end{corollary}

Analogous results were obtained recently by Espinoza~\cite{pre/Espinoza} and Savalev and Puzynina~\cite{th/Savelev2024}. Moreover, another recent result by Espinoza (unpublished) states that every proper unimodular shift space has $\ab$-stable return groups. Note that this last result also implies that the shift space of \cref{eg:non-stable} is $\ab$-stable, even though it is not eventually stable.

Another interesting application of \cref{t:welldoc} is the following corollary obtained when $\varphi$ is the identity on $F_\cA$.
\begin{corollary}
    A minimal shift space has stable return groups if and only if it has $\varphi$-welldoc for every morphism $\varphi$ onto a finite group.
\end{corollary}

In particular, we deduce the following, which generalizes in multiple ways a theorem by Balková et al. stating that Arnoux--Rauzy words have the welldoc property~\cite{Balkova2016}. 
\begin{corollary}\label{c:welldoc}
    Every minimal suffix-connected shift space, and in particular, every minimal dendric shift space, has $\varphi$-welldoc for every morphism $\varphi$ onto a finite group.
\end{corollary}

\section{Open questions}

Since this work is a first exploration of eventual stability of return groups, many questions remain unanswered. We would like to emphasize two that we believe are of particular interest.

\begin{enumerate}
    \item We recently showed that return words can be used to characterize dendricity, a purely combinatorial property~\cite{pre/Gheeraert2024}. Namely, a minimal shift space $X$ over $\cA$ is dendric if and only if $X$ has stable return groups and, for every $u \in \lang(X)$, $\card\ret_X(u) = \card\cA$. In light of this result and of the examples given in \cref{S:examples}, it is of interest to investigate the interactions between the algebraic properties of return words and some combinatorial aspects of shift spaces. For instance, can (eventual) $\varphi$-stability be used to characterize $k$-automatic shift spaces, or eventually dendric shift spaces? On the other hand, can (eventual) stability be understood from a combinatorial perspective?

    \item In \cref{S:closure properties}, we obtained results on closure properties of (eventual) stability with respect to derivation and substitutions. As for topological conjugacy, we know that stability is not preserved. Is eventual stability preserved by topological conjugacy? In other words, is it a dynamical property? Note that, by \cref{c:closed under substitution}, it is equivalent to ask whether it is closed under topological factorization.
\end{enumerate}

\bibliographystyle{abbrvnat}
\bibliography{biblio_propre}

\end{document}